\documentclass{elsarticle} \usepackage[T1]{fontenc}
\DeclareMathAlphabet{\monbb}{U}{bbold}{m}{n}
\usepackage{amsmath,amsfonts,amssymb, amsthm,url}
\usepackage{stmaryrd} 
\usepackage{graphics} \usepackage{color}
\usepackage{tikz}

\usepackage{marvosym,fancybox}
\newcommand\along{along}

\renewcommand{\epsilon}{\varepsilon}

\newcommand{\N}{\mathbb{N}} 
\newcommand{\Z}{\mathbb{Z}} \newcommand{\Q}{\mathbb{Q}}
\newcommand{\R}{\mathbb{R}}

\newcommand{\K}{\mathbb{K}} 
\newcommand{\U}{\mathbb{U}} 
 \newcommand{\A}{\mathcal{A}}
\newcommand{\B}{\mathcal{B}}
\newcommand{\az}{\mathcal{A}^{\mathbb{Z}}}

\newcommand{\bz}{\mathcal{B}^{\mathbb{Z}}}

\newcommand{\id}{\mathrm{Id}} 
 
 \newcommand{\s}{\sigma}

 \newcommand{\La}{\mathcal{L}}
\newcommand{\dd}{\delta} 
\newcommand{\e}{\epsilon}

\newcommand{\la}{\Lambda}

\newcommand{\gs}{\Sigma}

\newcommand{\tube}{D}

\newcommand{\aalk}{\textbf{A}_{\K}'}
\newcommand{\aaln}{\textbf{A}_{\N}'}

\newcommand{\alk}{\textbf{A}_{\K}} \newcommand{\aln}{\textbf{A}_{\N}}
 \newcommand{\bl}{\textbf{B}}
\newcommand{\blk}{\textbf{B}_{\K}} \newcommand{\bln}{\textbf{B}_{\N}}
\newcommand{\blz}{\textbf{B}_{\Z}}

\newcommand{\daalk}{\widetilde{\textbf{A}}_{\K}'}

\newcommand{\dalk}{\widetilde{\textbf{A}}_{\K}}
\newcommand{\daln}{\widetilde{\textbf{A}}_{\N}}

\newcommand{\dblk}{\widetilde{\textbf{B}}_{\K}}

\newcommand{\Fon}{\mathcal{F}} \newcommand{\Dir}{\mathcal{D}}

\newcommand{\Ap}{\mathcal{A}_P} 
\newcommand{\apz}{\mathcal{A}_P^{\Z}} 
\newcommand{\Up}{\U_P}
\newcommand{\Fp}{F_P} 
\newcommand{\Fc}{F_C} 
\newcommand\paraset{\mathcal{P}_0}

\newcommand{\siteFs}[2]{\langle #1,#2\rangle_{\s,F}} 
\newcommand{\site}[2]{\langle #1,#2\rangle} 

\newcommand{\cone}[2]{\mathfrak{C}_{#1}(#2)} 

\newcommand\subsubsubsection[1]{\paragraph{#1}}

\newcommand{\rmk}{\textbf{Note. }} 
\newcommand{\st}{\ |\ }

\theoremstyle{plain} \newtheorem{theorem}{Theorem}[section]
\newtheorem{prop}[theorem]{Proposition}
\newtheorem{corollary}[theorem]{Corollary}
\newtheorem{lemma}[theorem]{Lemma}

\theoremstyle{definition} 
\newtheorem{definition}{Definition}[section]

\theoremstyle{remark} 
\newtheorem{remark}{Remark}[section]
\newtheorem{ex}{Example}[section] 

\newcounter{factcount}[theorem]
\newcommand{\THMfont}[1]{{\sl #1}}
\newcommand{\fact}[1]{\refstepcounter{factcount} \vspace{0.1em}  \noindent {\sc Fact \thefactcount: \ }\THMfont{ #1}}
\newcommand{\bprf}[1][Proof:]{\begin{list}{}    {\setlength{\leftmargin}{0.5em} \setlength{\rightmargin}{0em}  \setlength{\listparindent}{1em}}   \item {\em \hspace{-1em}  #1  }}
\newcommand{\eprf}{\end{list}}

\newcommand{\bfactprf}{\bprf}

\newcommand{\efactprf}{ \hfill $\Diamond$~{\scriptsize {\tt Fact~\thefactcount}}\eprf} % %

\begin{document}
\begin{frontmatter} \title{Directional Dynamics along Arbitrary
Curves\\ in Cellular Automata}

  \author[lif]{M. Delacourt} \author[lif]{V. Poupet}
\author[latp]{M. Sablik\corref{cora}} \author[lama]{G. Theyssier}
\cortext[cora]{Corresponding author
(\url{sablik@latp.univ-mrs.fr})}

  \address[lif]{LIF, Aix-Marseille Universit\'e, CNRS, 39 rue
Joliot-Curie, 13\hspace{0.2em}013 Marseille, \rlap{France}}
\address[latp]{LATP, Universit\'e de Provence, CNRS, 39, rue Joliot
Curie, 13\hspace{0.2em}453 Marseille Cedex 13, \rlap{France}}
\address[lama]{LAMA, Universit\'e de Savoie, CNRS, 73\hspace{0.2em}376
Le Bourget-du-Lac Cedex, France}

  \begin{abstract} 
    This paper studies directional dynamics on one-dimensional cellular automata, a formalism
    previously introduced by the third author. The central idea is to study the
    dynamical behaviour of a cellular automaton through the conjoint action of
    its global rule (temporal action) and the shift map (spacial action):
    qualitative behaviours inherited from topological dynamics (equicontinuity,
    sensitivity, expansivity) are thus considered along arbitrary curves in
    space-time. The main contributions of the paper concern equicontinuous
    dynamics which can be connected to the notion of consequences of a word.  We
    show that there is a cellular automaton with an equicontinuous dynamics
    along a parabola, but which is sensitive along any linear direction.  We
    also show that real numbers that occur as the slope of a limit linear
    direction with equicontinuous dynamics in some cellular automaton are
    exactly the computably enumerable numbers.
  \end{abstract}

  \begin{keyword} cellular automata, topological dynamics, directional dynamics
  \end{keyword}
\end{frontmatter}

\section*{Introduction}

Introduced by J. von Neumann as a computational device, cellular
automata (CA) were also studied as a model of dynamical systems
\cite{Hedlund-1969}. G. A. Hedlund \textit{et al.} gave a
characterization of CA through their global action on configurations:
they are exactly the continuous and shift-commuting maps acting on the
(compact) space of configurations. Since then, CA were extensively
studied as discrete time dynamical systems for their remarkable
general properties (e.g., injectivity implies surjectivity) but also
through the lens of topological dynamics and deterministic chaos. With
this latter point of view, P. K\r{u}rka \cite{Kurka-1997} has proposed
a classification of 1D CA according to their equicontinuous properties
(see \cite{Sablik-Theyssier-2008} for a similar classification in
higher dimensions). As often remarked in the literature, the
limitation of this approach is to not take into account the
shift-invariance of CA: information flow is rigidly measured with
respect to a particular reference cell which does not vary with time
and, for instance, the shift map is considered as sensitive to initial
configurations.

One significant step to overcome this limitation was accomplished with the
formalism of \emph{directional dynamics} recently proposed by M. Sablik
\cite{Sablik-2008}. The key idea is to consider the action of the rule and that
of the shift simultaneously. CA are thus seen as $\Z^2$-actions (or
$\N\times\Z$-actions for irreversible rules). In \cite{Sablik-2008}, each
qualitative behaviour of K\r{u}rka's classification (equicontinuity, sensitivity,
expansivity) is considered for different linear correlations between the two
components of the $\Z^2$-action corresponding to different linear directions in
space-time. For a fixed direction the situation is similar to K\r{u}rka's
classification, but in \cite{Sablik-2008}, the classification scheme consists in
discussing what sets of directions support each qualitative behaviour.

The restriction to linear directions is natural, but \cite{Sablik-2008} asks
whether considering non-linear directions can be useful. One of the main points
of the present paper is to give a positive response to this question. We are
going to study each qualitative behaviour along arbitrary curves in space-time
and show that, in some CA, a given behaviour appears along some non-linear curve
but not along any linear direction. Another contribution of the paper is to give
a complete characterization of real numbers that can occur as (limit) linear
directions for equicontinuous dynamics.

Properties inherited from classical topological dynamics may have a concrete
interpretation when applied to CA. In particular, as remarked by P. K\r{u}rka,
the existence of equicontinuity points is equivalent to the existence of a
'wall', that is a word whose presence in the initial configuration implies an
infinite strip of \emph{consequences} in space-time (a portion of the lattice
has a determined value at each time step whatever the value of the configuration
outside the 'wall'). In our context, the connection between equicontinuous
dynamics and consequences of a word still apply but in a broader sense since we
consider arbitrary curves in space-time. The examples of dynamic behaviour along
non-trivial curves built in this paper will often rely on particular words whose
set of consequences have the desired shape.

Another way of looking at the notion of {consequences} of a word is to use the
analogy of information propagation and \emph{signals} already developed in the
field of classical algorithmics in CA \cite{Mazoyer-Terrier-1999}. From that
point of view, a word whose consequences follow a given curve in space-time can
be seen as a signal which is \emph{robust} to any pertubations from the
context. Thus, many of our results can be seen as constructions in a
non-standard algorithmic framework where information propagation must be robust
to any context. To achieve our results, we have developed general mechanisms to
introduce a form of robustness (counter technique, section~\ref{sec:examples}). We
believe that, besides the results we obtain, this technique is of some interest on its own.\\

After the next section, aimed at recalling useful definitions, the
paper is organized in four parts as follows:
\begin{itemize}
\item in section \ref{sec:theory}, we extend the theory of directional
  dynamics to arbitrary curves and prove a classification theorem
  analogue to that of \cite{Sablik-2008};
\item in section \ref{sec:examples}, we focus on equicontinuous
  dynamics and give constructions and construction tools; the main
  result is the existence of various CA where equicontinuous dynamics
  occur along some curve but not along others and particularly not
  along any linear direction;
\item in section \ref{sec:linear}, we focus on linear directions corresponding
  to equicontinuous dynamics; \cite{Sablik-2008} showed that the set of slopes
  of such linear directions is an interval (if not empty): we give a
  characterisation of real numbers that can occur as bounds of such intervals.
\item in section \ref{sec:limitations}, we give some negative results
  concerning possible sets of consequences of a word in CA; in
  particular, we show how the set of curves admitting equicontinuous
  dynamics is constrained in reversible CA.
\end{itemize}

% Quelques id\'ees:
% \begin{itemize}
% \item Historique classification dynamique.
% \item \'Etude d'un AC comme une $\Z^2$-action qui a d\'ebouch\'e sur
% une classification dans \cite{Sablik-2008}. Cela d\'ebouche sur
% certains prob\`emes ouverts.
% \item Un autre point de vue consiste \`a d\'efinir une dynamique
% suivant une courbe. Lien avec les d\'ependance d'une cellule.
% \item Recherche d'exemples: probl\`eme algorithmiques.
% \item On obtient des reponses aux questions ouvertes dans
% \cite{Sablik-2008}.
% \end{itemize}

%%%%%%%%%%%%%%%%%%%%%%%%%%%%%%%
%%%
%%%%%%%%%%%%%%%%%%%%%%%%%%%%%%%

\section{Some definitions}

\subsection{Space considerations}

\paragraph{Configuration space} Let $\A$ be a finite set and $\az$ the
{\em configuration space} of $\Z$-indexed sequences in $\A$. If $\A$
is endowed with the discrete topology, $\az$ is metrizable, compact
and totally disconnected in the product topology.  A compatible metric
is given by:
$$\forall x,y\in\az,\quad d_C(x,y)=2^{-\min\{|i|: x_i\ne y_i \ i\in\Z \}}.$$

Consider a not necessarily convex subset $\U\subset\Z$. For $x\in\az$,
denote $x_{\U}\in\A^{\U}$ the restriction of $x$ to $\U$. Given
$w\in\A^{\U}$, one defines the cylinder centered at $w$ by
$[w]_{\U}=\{x\in\az : x_{\U}=w\}$. Denote by $\A^{\ast}$ the set of
all finite sequences or {\em finite words} $w=w_0...w_{n-1}$ with
letters in $\A$; $|w|=n$ is the {\em length} of $w$. When there is no
ambiguity, denote $[w]_i=[w]_{\llbracket i,i+|w|-1\rrbracket}$.

\paragraph{Shift action} The {\em shift} map $\s:\az\rightarrow\az$ is
defined by $\s(x)_i=x_{i+1}$ for $x=(x_m)_{m\in\Z}\in\az$ and
$i\in\Z$. It is a homeomorphism of $\az$.

A closed and $\s$-invariant subset $\gs$ of $\az$ is called a {\em
  subshift}.  For $\U\subset\Z$ denote $\La_{\gs}(\U)=\{ x_{\U} :
x\in\gs \}$ the set of patterns centered at $\U$.  Since $\gs$ is
$\s$-invariant, it is sufficient to consider the words of length
$n\in\N$ for a suitable $n$. We denote $\La_{\gs}(n)=\{x_{\llbracket 0,n-1\rrbracket} :
x\in\gs\}$. The {\em language} of a subshift $\gs$ is defined by
$\La_{\gs}=\cup_{n\in\N}\La_{\gs}(n)$. By compactness, the language
characterizes the subshift.

A subshift $\gs \subseteq \az$ is {\em transitive} if given words $u,v
\in \La_{\gs}$ there is $w \in \La_{\gs}$ such that $uwv \in
\La_{\gs}$. It is {\em mixing} if given $u,v \in \La_{\gs}$ there is
$N\in \N$ such that $uwv \in \La_{\gs}$ for any $n \geq N$ and some $w
\in \La_{\gs}(n)$.

A subshift $\gs\subset\az$ is {\em specified } if there exists
$N\in\N$ such that for all $u,v\in\La_{\gs}$ and for all $n\geq N$
there exists a $\s$-periodic point $x\in\gs$ such that
$x_{\llbracket 0,|u|-1\rrbracket}=u$ and $x_{\llbracket n+|u|,n+|u|+|v|-1\rrbracket}=v$ (see
\cite{Denker-Grillenberger-Sigmund-1976l} for more details).

A subshift $\gs\subset\az$ is {\em weakly-specified } if there exists
$N\in\N$ such that for all $u,v\in\La_{\gs}$ there exist $n\leq N$ and
a $\s$-periodic point $x\in\gs$ such that $x_{\llbracket 0,|u|-1\rrbracket}=u$ and $x_{\llbracket n+|u|,n+|u|+|v|-1\rrbracket}=v$.

Specification (resp. weakly-specification) implies mixing
(resp. transitivity) and density of $\s$-periodic points. Let $\gs$ be
a weakly-specified mixing subshift. By compactness there exists $N \in
\N$ such that for any $x, y \in \gs$ and $i\in\N$ there exist $w \in
\La_{\gs}$, $|w|\leq N$, and $j\in\Z$ such that
$x_{\rrbracket-\infty,i\rrbracket}w\s^j(y)_{\llbracket i+|w|,\infty\llbracket} \in \gs$. If $\gs$ is
specified this property is true with $|w|=n$ and $n\geq N$.

\paragraph{Subshifts of finite type and sofic subshifts} A subshift
$\gs$ is of {\em finite type} if there exist a finite subset
$\U\subset\Z$ and $\mathcal{F}\subset\A^{\U}$ such that $x\in\gs$ if
and only if $\s^m(x)_{\U}\in\mathcal{F}$ for all $m\in\Z$. The
diameter of $\U$ is called an {\em order} of $\gs$.

A subshift $\gs'\subset\bz$ is {\em sofic} if it is the image of a
subshift of finite type $\gs\subset\az$ by a map $\pi: \az \to \bz$,
$\pi((x_i)_{i\in{\Z}})=(\pi(x_i))_{i\in \Z}$, where $\pi:\A\to\B$.

A transitive sofic subshift is weakly-specified and a mixing sofic subshift is
specified. For precise statements and proofs concerning sofic subshifts
and subshifts of finite type see~\cite{Lind-Marcus-1995l}
or~\cite{Kitchens-1998l}.

\subsection{Time considerations}

\subsubsection*{Cellular automata}

% A \emph{cellular automaton} (CA) is a dynamical system $(\A^\Z, F)$
% characterized by a function $F$ from a set of configurations $\A^\Z$
% into itself that can be defined locally as follows. There exists a
% finite rule $\delta$ that is applied uniformly and synchronously on
% the configuration: given a finite segment $N = [l, r]\subseteq \Z$
% (\emph{neighborhood} of the automaton) and a function $\delta: \A^N
% \rightarrow \A$ (\emph{local rule}), the \emph{global
%   function} $F$ is defined as:
% \[
% F : \left\{
%   \begin{array}{rcl}
%     \A^\Z & \rightarrow & \A^\Z \\
%     x & \mapsto & F(x), \quad \forall i\in \Z, F(x)_i = \delta(x_{i+l}, \ldots, x_{i+r})
%   \end{array}
% \right.
% \]
% 
% The \emph{radius} of a cellular automaton is $r(F) = max\{|x|, x\in
% N\}$.
% 
% Gustav A. Hedlund showed that Cellular automata are exactly the
% dynamical systems $(\A^\Z, F)$ for which $F$ is a continuous mapping
% from $\A^\Z$ into itself that commutes with the shift $\sigma$.

A {\em cellular automaton} (CA) is a dynamical system $(\A^\Z, F)$ defined by a local rule which acts uniformly and
synchronously on the configuration space. That is, there are a finite
segment or {\em neighborhood} $\U\subset\Z$ and a {\em local rule}
$\overline{F}:\A^{\U}\rightarrow\A$ such that
$F(x)_m=\overline{F}((x_{m+u})_{u\in\U})$ for all $x\in\az$ and
$m\in\Z$. The {\em radius} of $F$ is $r(F)=\max\{|u|: u\in\U\}$. By
Hedlund's theorem \cite{Hedlund-1969}, a {\em cellular automaton} is
equivalently defined as a pair $(\az,F)$ where $F:\az\to\az$ is a
continuous function which commutes with the shift.

\paragraph{Considering the past: bijective CA} When the CA is bijective, since
$\az$ is compact, $F^{-1}$ is also a continuous function which commutes with
$\s$. By Hedlund's theorem, $(\az,F^{-1})$ is then also a CA (however the radius
of $F^{-1}$ can be much larger than that of $F$). In this
case one can study the $\Z$-action $F$ on $\az$ and not only $F$ as an
$\N$-action. This means that we can consider positive (future) and negative
(past) iterates of a configuration.

Thus, if the CA is bijective, we can study the dynamic of the CA as an
$\N$-action or a $\Z$-action. In the general case, we consider the
$\K$-action of a CA where $\K$ can be $\N$ or $\Z$.

\subsection{A CA as a $\Z\times\K$-action}

\paragraph{Space-time diagrams}Let $(\az,F)$ be a CA, since $F$
commutes with the shift $\s$, we can consider the $\Z\times\K$-action
$(\s,F)$. For $x\in\az$, we denote by
${\siteFs{m}{n}(x)=\bigl(\s^m\circ F^n(x)\bigr)_{0}}$ the {\em color of
  the site} $\site{m}{n}\in\Z\times\K$ generated by $x$.

Adopting a more geometrical point of view, we also refer to this
coloring of $\Z\times\K$ as the \emph{space-time diagram} generated by
$x$. 

% \TODO{Expliquer le lien avec les diagrammes espace temps}

\paragraph{Region of consequences}Let $X \subset\az$ be any set of
configurations. We define the {\em region of consequences} of $X$ by:
$$\cone{F}{X}=\left\{\site{m}{n} \in\Z\times\K: \forall x,y\in X \textrm{ one has }\siteFs{m}{n}(x)=\siteFs{m}{n}(y)\right\}.$$  

This set corresponds to the sites that are fixed by all $x\in X$ under
the $\Z\times\K$-action $(\s,F)$, or equivalently, sites which are
identically colored in all space-time diagrams generated by some $x\in
X$. The main purpose of this article is to study this set and make links
with notions from topological dynamics.

Let $(\az,F)$ be a CA of neighborhood $\U=\llbracket r,s\rrbracket$ and let
$u\in\A^{+}$. An example of such set $X$ that will be used throughout
the paper is $[u]_0$. Trivially, one has (see figure~\ref{fig:wordcons})
\[
	\left\{ \site{m}{n}: nr\leq m < |u| - ns \right\} \subseteq
	\cone{F}{[u]_0} \subseteq
	\left\{\site{m}{n} : -nr \leq m < |u| + ns \right\}
\]
In the sequel, we often call $\cone{F}{[u]_0}$ the
\emph{cone of consequences of $u$}. Note that the inclusions above do
not tell whether $\cone{F}{[u]_0}$ is finite or infinite.

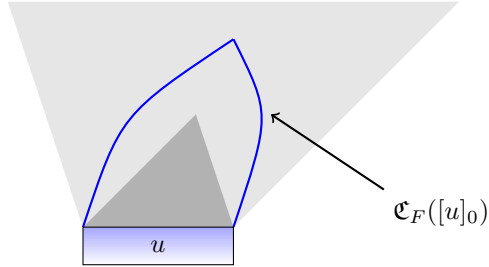
\begin{figure}
  \centering
  \begin{tikzpicture}
    \fill[gray!20!white] (0,.5)--(-1,3.5)--(5,3.5)--(2,.5)--cycle;
    \fill[fill=gray!60!white] (0,.5)--(2,.5)--(1.5,2)--cycle;
    \draw[blue, thick] (0,.5)..controls (.5, 2)..(2,3);
    \draw[blue, thick] (2,.5)..controls (2.5, 2)..(2,3);
    \draw (4,1) node[below right] {$\cone{F}{[u]_0}$};
    \draw[->, thick] (4,1)--(2.5,2);
    \draw[bottom color=white, top color=blue!40!white] (0,0)--(2,0)--(2,.5)--(0,.5)--cycle;
    \draw (1,.25) node {$u$};
  \end{tikzpicture}
  \caption{Consequences of a word $u$.}
  \label{fig:wordcons}
\end{figure}

%%%%%%%%%%%%%%%%%%%%%%%%%%%%%%%
%%%
%%%%%%%%%%%%%%%%%%%%%%%%%%%%%%%

\section{Dynamics along an arbitrary curve}
\label{sec:theory}

In this section, we define sensitivity to initial conditions along a
curve and we establish a connection with cones of consequences.  What
we call a curve is simply a map $h:\K\rightarrow\Z$ giving a position
in space for each time step. Such $h$ can be arbitrary in the
following definitions, but later in the paper we will put restrictions
on them to adapt to the local nature of cellular automata.

\subsection{Sensitivity to initial conditions along a curve}

Let $\gs$ be a subshift of $\az$ and assume $\K=\N$ or $\Z$.

Let $x\in\gs$, $\e>0$ and $h:\K\to\Z$. The {\em ball (relative to
  $\gs$) centered at $x$ of radius $\e$} is given by
$B_{\gs}(x,\e)=\{y\in\gs : d_C(x,y)<\e\}$ and the {\em tube \along{}
  $h$ centered at $x$ of radius $\e$} is (see figure~\ref{fig:tubr}):
$$
\tube^{h}_{\gs}(x,\e,\K)=\{y\in\gs : d_C(\s^{h(n)}\circ
F^n(x),\s^{h(n)}\circ F^n(y))<\e , \forall n\in\K \}.
$$

Notice that one can define a distance
$D(x,y)=\sup(\{d_C(\s^{h(n)}\circ F^n(x),\s^{h(n)}\circ F^n(y)) :
\forall n\in\N\})$ for all $x,y\in\gs$. The tube
$\tube^{h}_{\gs}(x,\e,\K)$ is then nothing else than the open ball of
radius $\e$ centered at $x$.

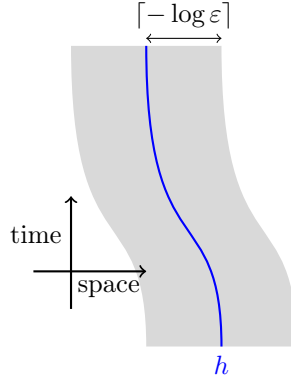
\begin{figure}
  \centering
  \begin{tikzpicture}
    \fill[fill=gray!30] (1,-1)..controls (1, 1) and (0,0)..(0,3)--(2,3) .. controls (2,0) and (3,1) .. (3,-1)--cycle;
    \draw[->, thick] (-.5,0)--(1,0);
    \draw[->, thick] (0,-.5)--(0,1);
    \draw (.5,0) node[below] {space};
    \draw (0,.5) node[left] {time};
    \draw[blue, thick] (2,-1)..controls (2, 1) and (1,0)..(1,3);
    \draw[blue] (2,-1) node[below] {$h$};
    \draw[<->] (1,3.1)--(2,3.1);
    \draw (1.5,3.1) node[above] {$\lceil-\log\e\rceil$};
%     \fill[fill=gray!30] (3,2)--(4,2)--(4,1.5)--(3,1.5)--cycle;
%     \draw (4,1.75) node[right] {: $F^n(x)$ and $F^n(y)$ match};
  \end{tikzpicture}
  \caption{Tube along $h$ of width $\e$ centered at $x$. The gray
    region is where $F^n(x)$ and $F^n(y)$ must match.}
  \label{fig:tubr}
\end{figure}

If the CA is bijective, one can assume that $\K=\Z$.

\begin{definition} Assume $\K=\N$ or $\Z$. Let $(\az,F)$ be a CA,
$\gs\subset\az$ be a subshift and $h:\K\to\Z$.
  \begin{itemize}
  \item The set $Eq^{h}_{\K}(\gs,F)$ of {\em $(\K,\gs)$-equicontinuous
points \along{} $h$} is defined by
    $$x\in Eq^{h}_{\K}(\gs,F) \Longleftrightarrow \forall \e>0, \exists\dd>0, \ B_{\gs}(x,\dd)\subset \tube^{h}_{\gs}(x,\e,\K).$$

% \item $(\az,F)$ {\em has $(\K,\gs)$-equicontinuous points \along{}
% $h$} if $Eq^{h}_{\K}(\gs,F)\ne\emptyset$.

  \item $(\az,F)$ is (uniformly) {\em $(\K,\gs)$-equicontinuous \along{} $h$} if
    $$\forall \e>0, \exists \dd>0, \forall x\in\gs, \ B_{\gs}(x,\dd)\subset \tube^{h}_{\gs}(x,\e,\K).$$

  \item $(\az,F)$ is {\em $(\K,\gs)$-sensitive \along{} $h$} if
    $$\exists \e>0, \forall \dd>0, \forall x\in\gs , \ \exists y\in B_{\gs}(x,\dd)\smallsetminus \tube^{h}_{\gs}(x,\e,\K).$$

  \item $(\az,F)$ is {\em $(\K,\gs)$-expansive \along{} $h$} if
    $$\exists \e>0, \forall x\in\gs, \ \tube^{h}_{\gs}(x,\e,\K)=\{x\}.$$
  \end{itemize}

  Since the domain of a CA is a two sided fullshift, it is possible to
break up the concept of expansivity into right-expansivity and
left-expansivity. The intuitive idea is that `information'' can move by
the action of a CA to the right and to the left.

  \begin{enumerate}
  \item[$\bullet$] $(\az,F)$ is {\em $(\K,\gs)$-right-expansive \along{} $h$} if there exists $\e>0$ such that
$\tube_{\gs}^{h}(x,\e,\K)\cap \tube_{\gs}^{h}(y,\e,\K)=\emptyset$ for
all $x,y\in\gs$ such that $x_{\llbracket 0,+\infty\llbracket}\ne y_{\llbracket 0,+\infty\llbracket}$.

  \item[$\bullet$] $(\az,F)$ is {\em $(\K,\gs)$-left-expansive \along{} $h$} if there exists $\e>0$ such that
$\tube_{\gs}^{h}(x,\e,\K)\cap \tube_{\gs}^{h}(y,\e,\K)=\emptyset$ for
all $x,y\in\gs$ such that $x_{\rrbracket-\infty,0\rrbracket}\ne y_{\rrbracket-\infty,0\rrbracket}$.
  \end{enumerate}

  Thus the CA $(\az,F)$ is $(\K,\gs)$-expansive \along{}~$h$ if it is
  both $(\K,\gs)$-left-expansive and $(\K,\gs)$-right-expansive
  \along{}~$h$.
\end{definition}

For $\alpha\in\R$, define:
$$\begin{array}{llll}
  h_{\alpha} : & \K & \longrightarrow & \Z \\ & n & \longmapsto &
\lfloor \alpha n\rfloor.
\end{array}$$ Thus, dynamics \along{} $\alpha$ introduced
in~\cite{Sablik-2008} correspond to dynamics along $h_{\alpha}$
defined in this paper.

\subsection{Blocking words for functions with bounded variation}

To translate equicontinuity concepts into space-time diagrams
properties, we need the notion of {\em blocking word} \along{}
$h$. The {\em wall} generated by a blocking word can be interpreted as
a particle which has the direction~$h$ and kills any information
coming from the right or the left. For that we need that the variation
of the function $h$ is bounded.

\begin{definition}
The {\em set of functions with bounded variation} is defined by:
  \[\Fon=\left\{h:\K\to\Z : \exists M>0,\ \forall n\in\K,\  |h(n+1)-h(n)|\leq M\right\}.\]
\end{definition}

Note that $\Fon$ depends on $\K$, but we will never make this explicit
and the context will always make this notation unambiguous in the
sequel.

\begin{definition} Assume $\K=\N$ or $\Z$. Let $(\az,F)$ be a CA with
  neighborhood $\U=\llbracket r,s\rrbracket$ (same neighborhood for $F^{-1}$ if $\K=\Z$).

  Let $\gs\subset\az$ be a subshift, $h\in\Fon$, $e\in\N$ such that
  \[
	e > \max_{n\in\K}(|h(n+1)-h(n)|+s,|h(n+1)-h(n)|-r)
	\]
	and $u\in\La_{\gs}$ with $|u|\geq e$. The
  word $u$ is a {\em $(\K,\gs)$-blocking word \along{} $h$ and width
$e$} if there exists a $p\in\Z$ such that (see figure~\ref{fig:blockw}):
  $$\cone{F}{\gs\cap[u]_p}\supset\left\{ \site{m}{n} \in\Z\times\K: h(n)\leq m < h(n)+e \right\}.$$
\label{def:blockingword}
\end{definition}
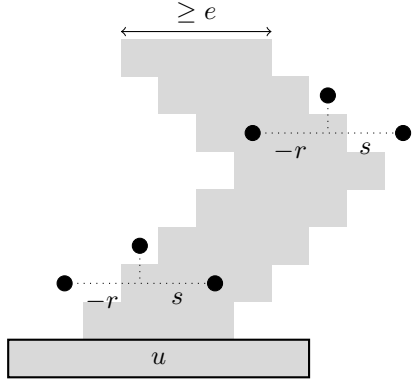
\begin{figure}
  \centering
  \begin{tikzpicture}[scale=.5]
    \foreach \i in {0,...,3}%
    \fill[fill=gray!30] (\i,\i)-- +(4,0)-- +(4,1)-- +(0,1)--cycle;%
    \foreach \i in {0,...,3}%
    \fill[fill=gray!30] (-\i,\i)++(4,4)--+(4,0)--+(4,1)--+(0,1)--cycle;%
    \draw[thick,fill=gray!30] (-2,-1)-- +(8,0)-- +(8,1)-- +(0,1)--cycle;%
    \draw[<->] (1,8)++(0,.2)--+(4,0);
    
    \newcommand\localrule[1]{%
      \draw[fill=black] #1++(.5,.5) circle (.2cm);
      \draw[fill=black] #1++(-1.5,-.5) circle (.2cm);
      \draw[fill=black] #1++(2.5,-.5) circle (.2cm);
%       \draw #1--+(1,0)--+(1,1)--+(0,1)--cycle;
%       \draw #1++(-2,-1)--+(1,0)--+(1,1)--+(0,1)--cycle;
%       \draw #1++(2,-1)--+(1,0)--+(1,1)--+(0,1)--cycle;
      \draw[dotted] #1++(-1.5,-.5)--+(4,0);
      \draw[dotted] #1++(.5,.5)--+(0,-1);
      \path #1++(-1.5,-.5)--+(2,0) node[sloped,below,midway] {$-r$};
      \path #1++(.5,-.5)--+(2,0) node[sloped,below,midway] {$s$};
    }
    
    \localrule{(1,2)}
    \localrule{(6,6)}
    \path (1,8)++(0,.2) -- +(4,0) node[sloped,above,midway] {$\geq e$};
    \draw (2,-.5) node {$u$};
  \end{tikzpicture}
  \caption{Blocking word $u$ of width $e$ for a CA of neighborhood $\llbracket r,s\rrbracket$. The
    gray region represents the consequences of $u$.}
  \label{fig:blockw}
\end{figure}

The evolution of a cell $i\in\Z$ depends on the cells
$\llbracket i+r,i+s\rrbracket$. Thus, due to condition on $e$, it is easy to deduce that
if $u$ is a $(\K,\gs)$-blocking word \along{} $h$ and width $e$, then
for all $j\in\Z$, $x,y\in [u]_j\cap\gs$ such that
$x_{\rrbracket-\infty,j\rrbracket}=y_{\rrbracket-\infty,j\rrbracket}$ and $n\in\K$ one has
$F^n(x)_i=F^n(y)_i$ for $i\leq h(n)+p+e+j$.  Similarly for all $x,y\in
[u]_j\cap\gs$ such that $x_{\llbracket j,\infty\llbracket}=y_{\llbracket j,\infty\llbracket}$, one has
$F^n(x)_i=F^n(y)_i$ for all $i \geq h(n)+p$.  Intuitively, no
information can cross the wall \along{} $h$ and width $e$ generated by
the $(\K,\gs)$-blocking word.

The proof of the classification of CA given in \cite{Kurka-1997} can
be easily adapted to obtain a characterization of CA which have
equicontinuous points \along{} $h$.

\begin{prop}\label{bloc-equicont} Assume $\K=\N$ or $\Z$. Let
$(\az,F)$ be a CA, $\gs\subset\az$ be a transitive subshift and
$h\in\Fon$. The following properties are equivalent:
  \begin{enumerate}
  \item $(\az,F)$ is not $(\K,\gs)$-sensitive \along{} $h$;

  \item $(\az,F)$ has a $(\K,\gs)$-blocking word \along{} $h$;

  \item $Eq^{h}_{\K}(\gs,F)\ne\emptyset$ is a $\s$-invariant dense
$G_{\dd}$ set.
  \end{enumerate}
\end{prop}
\begin{proof} Let $\U=\llbracket r,s\rrbracket$ be a neighborhood of $F$ (and also of
$F^{-1}$ if $\K=\Z$).

  $(1)\Rightarrow(2)$ Let
$e\geq\max_{n\in\K}(|h(n+1)-h(n)|+1+s,|h(n+1)-h(n)|+1-r)$. If
$(\az,F)$ is not $(\K,\gs)$-sensitive \along{} $h$, then there exist
$x\in\gs$ and $k,p\in\N$ such that for all $y\in\gs$ verifying
$x_{\llbracket 0,k\rrbracket}=y_{\llbracket 0,k\rrbracket}$ one has:
  $$\forall n\in\K,\ \s^{h(n)}\circ F^n(x)_{\llbracket p,p+e-1\rrbracket}=\s^{h(n)}\circ F^n(y)_{\llbracket p,p+e-1\rrbracket}.$$
  Thus $x_{\llbracket 0,k\rrbracket}$ is a $(\K,\gs)$-blocking word \along{} $h$ and
width $e$.

  $(2)\Rightarrow(3)$ Let $u$ be a $(\K,\gs)$-blocking word \along{}
$h$. Since $(\gs,\s)$ is transitive, then there exists $x\in\gs$
containing an infinitely many occurrences of $u$ in positive and
negative coordinates. Let $k\in\N$. There exists $k_1\geq k$ and
$k_2\geq k$ such that
$x_{\llbracket-k_1,-k_1+|u|-1\rrbracket}=x_{\llbracket k_2,k_2+|u|-1\rrbracket}=u$. Since $u$ is a
$(\K,\gs)$-blocking word \along{} $h$, for all
$y\in\gs$ such that $y_{\llbracket -k_1,k_2+|u|-1\rrbracket}=x_{\llbracket -k_1,k_2+|u|-1\rrbracket}$ one
has
  $$\s^{h(n)}\circ F^n(x)_{\llbracket-k,k\rrbracket}=\s^{h(n)}\circ F^n(y)_{\llbracket-k,k\rrbracket} \quad \forall n\in\K.$$
  One deduces that $x\in Eq^{h}_{\K}(\gs,F)$.

  Moreover, since $\gs$ is transitive, the subset of points in $\gs$
containing infinitely many occurrences of $u$ in positive and negative
coordinates is a $\s$-invariant dense $G_{\dd}$ set of $\gs$.

  $(3)\Rightarrow(1)$ Follows directly from definitions.
\end{proof}

\begin{remark} When $\gs$ is not transitive one can show that any
$(\K,\gs)$-equicontinuous point \along{} $h$ contains a
$(\K,\gs)$-blocking word \along{} $h$. Reciprocally, a point $x\in\gs$
containing infinitely many occurrences of a $(\K,\gs)$-blocking word
\along{} $h$ in positive and negative coordinates is a
$(\K,\gs)$-equicontinuous point \along{} $h$. However, if $\gs$ is not
transitive, the existence of a $(\K,\gs)$-blocking word does not imply
that one can repeat it infinitely many times.
\end{remark}

\subsection{A classification following a curve}

Thanks to Proposition~\ref{bloc-equicont} it is possible to establish
a classification as in~\cite{Kurka-1997}, but following a given curve.

\begin{theorem}\label{classification} Assume $\K=\N$ or $\Z$. Let
$(\az,F)$ be a CA, $\gs\subset\az$ be a transitive subshift and
$h\in\Fon$. One of the following cases holds:
  \begin{enumerate}
  \item $Eq^{h}_{\K}(\gs,F)=\gs$ $\Longleftrightarrow$ $(\az,F)$ is
$(\K,\gs)$-equicontinuous \along{} $h$;
  \item $\emptyset\ne Eq^{h}_{\K}(\gs,F)\ne\gs$ $\Longleftrightarrow$
$(\az,F)$ is not $(\K,\gs)$-sensitive \along{} $h$
$\Longleftrightarrow$ $(\gs,F)$ has a $(\K,\gs)$-blocking word \along{} $h$;

  \item $(\az,F)$ is $(\K,\gs)$-sensitive \along{} $h$ but is not
$(\K,\gs)$-expansive \along{} $h$;

  \item $(\az,F)$ is $(\K,\gs)$-expansive \along{} $h$.
  \end{enumerate}
\end{theorem}
\begin{proof}

  First we prove the first equivalence. From definitions we deduce
  that if $(\az,F)$ is $(\K,\gs)$-equicontinuous \along{} $h$ then
  $Eq^{h}_{\K}(\gs,F)=\gs$. In the other direction, consider the
  distance $D(x,y)=\sup(\{d_C(\s^{h(n)}\circ F^n(x),\s^{h(n)}\circ
  F^n(y)) : \forall n\in\N\})$ mentionned earlier.  $Eq^h_{\K}(\gs,F)$
  is the set of equicontinuous points of the function
  $\id:(\gs,d_C)\to(\gs,D)$. By compactness, if this function is
  continuous on $\gs$, then it is uniformly continuous. One deduces
  that $(\az,F)$ is $(\K,\gs)$-equicontinuous \along{} $h$.

  The second equivalence and the classification follow directly from
Proposition~\ref{bloc-equicont}.
\end{proof}

%%%%%%%%%%%%%%%%%%%%%%%%%%%%%%%
%%%
%%%%%%%%%%%%%%%%%%%%%%%%%%%%%%%

\subsection{Sets of curves with a certain kind of dynamics}

% \subsection{Sets of curves and their relations}

We are going to study the sets of curves along which a certain kind of
dynamics happens. We obtain a classification similar at the
classification obtained in~\cite{Sablik-2008} but not restricted to
linear directions.

\begin{definition}
  Assume $\K=\N$ or $\Z$. Let $(\az,F)$ be a CA and $\gs$ be a
  subshift. We define the following sets of curves.
  \begin{itemize}
  \item Sets corresponding to topological equicontinuous properties:
    \begin{eqnarray*} \alk(\gs,F)&=&\{h\in\Fon : Eq^{h}_{\K}(\gs,F)
      \ne\emptyset\},\\ \textrm{and }\aalk(\gs,F)&=&\{h\in\Fon :
      Eq^{h}_{\K}(\gs,F)=\gs\}.
    \end{eqnarray*} One has $\aalk(\gs,F)\subset\alk(\gs,F).$
  
  \item Sets corresponding to topological expansive properties:
    \begin{eqnarray*} \blk(\gs,F)&=&\{h\in\Fon : (\az,F) \textrm{ is
        $(\K,\gs)$-expansive \along{} }h \},\\
      \blk^{r}(\gs,F)&=&\{h\in\Fon :
      (\az,F) \textrm{ is $(\K,\gs)$-right-expansive \along{} }h\},\\
      \textrm{and }\quad \blk^{l}(\gs,F)&=&\{h\in\Fon : (\az,F)
      \textrm{ is $(\K,\gs)$-left-expansive \along{} }h \}.
    \end{eqnarray*} One has
    $\blk(\gs,F)=\blk^{r}(\gs,F)\cap\blk^{l}(\gs,F).$
  \end{itemize}
\end{definition}
\begin{remark} The set of directions which are $(\K,\gs)$-sensitive is
  $\Fon\setminus\alk(\gs,F)$, so it is not necesary to study this set.
\end{remark}

% \TODO{recycler ce qui suit.}
% \begin{center}
%   \tiny
%   \begin{minipage}{.8\linewidth}
%     \noindent Moreover, for $h\in\alz(\gs,F)$
%     (resp. $h\in\aalz(\gs,F)$), one considers
%   $$\begin{array}{cccc}h_1: &\N&\to&\Z \\ &n&\mapsto&h(n) \end{array}\qquad\textrm{ and }\qquad\begin{array}{cccc}h_2: &\N&\to&\Z \\ &n&\mapsto&h(-n) \end{array}$$ and we have  $h_1\in\aln(\gs,F)$ and $h_2\in\aln(\gs,F^{-1})$ ($h_1\in\aaln(\gs,F)$ and $h_2\in\aaln(\gs,F^{-1})$).

%   \noindent Moreover, if $F$ is reversible,
%   $\left(\bln^{l}(\gs,F)\cap\bln^{r}(\gs,F^{-1})\right)\cup\left(\bln^{r}(\gs,F)\cap\bln^{l}(\gs,F^{-1})\right)\subset\blz(\gs,F).$
% \end{minipage}
% \end{center}

Let $\Dir=\{h_{\alpha}:\alpha\in\R\}$. In~\cite{Sablik-2008}, we
consider the sets $\dalk(\gs,F)=\alk(\gs,F)\cap\Dir$,
$\daalk(\gs,F)=\aalk(\gs,F)\cap\Dir$ and
$\dblk(\gs,F)=\blk(\gs,F)\cap\Dir$.

The remaining part of the section aims at generalizing this classification to
$\Fon$, the set of curves with bounded variation.

\subsection{Equivalence and order relation on $\Fon$}

\begin{definition} Let $h,k\in\Fon$.

  Put $h\precsim k$ if there exists $M>0$ such that $h(n)\leq k(n)+M$
for all $n\in\K$.

  Define $h\sim k$ if there exists $M>0$ such that $k(n)-M\leq
h(n)\leq k(n)+M$ for all $n\in\K$.

  Define $h\prec k$ if $h\precsim k$ and $h\nsim k$.
\end{definition}

It is easy to verify that $\precsim$ is an semi-order relation on
$\Fon$ and $\sim$ is the equivalence relation on $\Fon$ associated to
$\precsim$.

\begin{prop}
	\label{prop:order_on_F}
	Let $(\az,F)$ be a CA, $\gs$ be a transitive subshift and
	$h,k\in\Fon$.
	\begin{itemize}
		\item If $h\precsim k$ then $h\in\blk^{r}(\gs,F)$ implies
  	$k\in\blk^{r}(\gs,F)$ and $k\in\blk^{l}(\gs,F)$ implies
  	$h\in\blk^{l}(\gs,F)$.

		\item If $h\sim k$ then $h\in\aalk(\gs,F)$ (resp. in $\alk(\gs,F)$,
  	$\blk^{l}(\gs,F)$, $\blk^{r}(\gs,F)$, $\blk(\gs,F)$) implies
  	$k\in\aalk(\gs,F)$ (resp. in $\alk(\gs,F)$, $\blk^{l}(\gs,F)$,
  	$\blk^{r}(\gs,F)$, $\blk(\gs,F)$).
	\end{itemize}
\end{prop}

\begin{proof}
  Straightforward.
\end{proof}

\subsection{Properties of $\alk(\gs,F)$}

The next proposition shows that $\alk(\gs,F)$ can be seen as a
``convex'' set of curves.

\begin{prop}
	\label{convexe_al}
		Let $(\az,F)$ be a CA and
		$\gs\subset\az$ be a transitive subshift. If
		$h',h''\in\aln(\gs,F)$ then for all $h\in\Fon$ which verifie
		$h'\precsim h \precsim h''$, one has $h\in\aln(\gs,F)$.
\end{prop}
\begin{proof}
	If $h' \sim h''$, by Proposition \ref{prop:order_on_F}, there is nothing to prove. Assume that $h' \prec h''$, we can consider two
	$(\N,\gs)$-blocking words $u'$ and $u''$ \along{} $h'$ and $h''$
	respectively. So there exist
	$e',e''\geq\max_{n\in\N}(|h''(n+1)-h''(n)|+1+s,|h'(n+1)-h'(n)|+1-r)$,
	$p'\in \llbracket 0,|u'|-e'\rrbracket$ and $p''\in \llbracket 0,|u''|-e''\rrbracket$ such that for all
	$x',y'\in[u']_0\cap\gs$, for all $x'',y''\in[u'']_0\cap\gs$ and for
	all $n\in\N$:
  \begin{eqnarray*} \s^{h'(n)}\circ
		F^n(x')_{\llbracket p',p'+e'-1\rrbracket}&=&\s^{h'(n)}\circ F^n(y')_{\llbracket p',p'+e'-1\rrbracket}\\
		\textrm{ and }\quad \s^{h''(n)}\circ
		F^n(x'')_{\llbracket p'',p''+e''-1\rrbracket}&=&\s^{h''(n)}\circ
		F^n(y'')_{\llbracket p'',p''+e''-1\rrbracket}.
  \end{eqnarray*}

  Since $\gs$ is transitive, there exists $w\in\La_{\gs}$ such that
$u=u'wu''\in\La_{\gs}$. For all $x,y\in [u]_0\cap\gs$ and for all
$n\in\N$ one has:
  $$F^n(x)_{\llbracket p'+h'(n), |u'| + p'' + e'' - 1 + h''(n)\rrbracket} = F^n(y)_{\llbracket p'+h'(n),|u'|+p''+e''-1+h''(n)\rrbracket}.$$
  This implies that $u$ is a $(\N,\gs)$-blocking word \along{} $h$ for
all $h\in\Fon$ which verifies $h'\precsim h \precsim h''$.
\end{proof}

\begin{definition} Let $(\az,F)$ be a CA and $\gs$ be a
subshift. $(\az,F)$ is {\em $\gs$-nilpotent} if the {\em $\gs$-limit
set} defined by
  $$\Lambda_{\gs}(F)=\cap_{n\in\N}\overline{\cup_{m\geq n}F^m(\gs)},$$ is finite. By compactness, in this case there exists $n\in\N$ such that $F^n(\gs)=\Lambda_{\gs}(F)$.
\end{definition} We observe that in general $\gs$ is not
$F$-invariant.

\begin{prop}\label{voisinage_al} Let $(\az,F)$ be a CA of neighborhood
$\U=\llbracket r,s\rrbracket$ and $\gs\subset\az$ be a weakly-specified subshift. If
there exists $h\in\alk(\gs,F)$ such that $h\prec h_{-s}$ or
$h_{-r}\prec h$ then $(\az,F)$ is $\gs$-nilpotent, thus
$\aln(\gs,F)=\Fon$.
\end{prop}
\begin{proof} Let $u$ be a $(\N,\gs)$-blocking word \along{}
$h\in\Fon$ with $h_{-r}\prec h$ and width $e$. There exists $p\in
\llbracket 0,|u|-e\rrbracket$ such that $$\forall n\in\N, \forall x,y\in[u]_0\cap\gs,
F^n(x)_{\llbracket h(n)+p, h(n)+p+e-1\rrbracket}=F^n(y)_{\llbracket h(n)+p, h(n)+p+e-1\rrbracket}.$$

  Let $z\in\gs\cap [u]_0$ be a $\s$-periodic configuration. The
sequence $(F^n(z))_{n\in\N}$ is ultimately periodic of preperiod $m$
and period $t$. Denote by $\gs'$ the subshift generated by
$(F^n(z))_{n\in\llbracket m,m+t-1\rrbracket}$, $\gs'$ is finite since $F^n(z)$ is a
$\s$-periodic configuration for all $n\in\N$. Let $q$ be the order of
the subshift of finite type $\gs'$.

  Since $\gs$ is a weakly-specified subshift, there exists $N\in\N$
such that for all $w,w'\in\La_{\gs}$ there exist $k\leq N$ and
$x\in\gs$ a $\s$-periodic point such that $x_{\llbracket 0,|w|-1\rrbracket}=w$ and
$x_{\llbracket k+|w|,k+|w|+|w'|-1\rrbracket}=w'$. Let $n\in\N$ be such that $|u|+N-rn+q\leq
h(n) + p + e$ (it is possible since $h_{-r}\prec h$). We want to prove that
$F^n(\gs)\subset\gs'$.

  The set $\llbracket rn,sn\rrbracket$ is a neighborhood of $(\az,F^n)$. Let
$v\in\La_{\gs}((s-r)n+q)$. There exist $x\in\gs$ and $k\leq N$, such
that $x_{\rrbracket -\infty,|u|-1\rrbracket}=z_{\rrbracket-\infty,|u|-1\rrbracket}$ and
$x_{\llbracket |u|+k,|u|+k+|v|-1\rrbracket}=v$. Since $u$ is a $(\N,\gs)$-blocking word
\along{} $h$, the choice of $n$ implies that
$F^n(x)_{\llbracket|u|+N-rn,|u|+N-rn+q-1\rrbracket}=F^n(z)_{\llbracket|u|+N-rn,|u|+N-rn+q-1\rrbracket}$. One
deduces that the image of the function
$\overline{F^n}:\La_{\gs}(\llbracket rn,sn+q\rrbracket)\to \A^q$ is contained in
$\La_{\gs'}(q)$. One deduces that $F^n(\gs)\subset\gs'$ so $(\az,F)$
is $\gs$-nilpotent which implies that $\aln(\gs,F)=\Fon$.

  The same proof holds for $h\prec h_{-s}$.
\end{proof}

\begin{remark} If moreover $\gs$ is specified, the same proof shows
that there exists $\A_{\infty}\subset\A$ such that
$\la_F(\gs)=\{^{\infty}a^{\infty} : a\in\A_{\infty}\}$.
\end{remark}

\begin{ex}[Importance of the specification hypothesis in Proposition~\ref{voisinage_al}] 
  Consider $(\{0,1\}^{\az},F)$ such that $F(x)_i=x_{i-1}\cdot
  x_{i}\cdot x_{i+1}$. Let $f^-,f^+\in\Fon$ such that $f^-\precsim
  h_{-1}$ and $h_{1}\precsim f^+$. Define $\gs_{f^-,f^+}$ as the
  maximal subshift such that $\La_{\gs_f}\cap\{10^m 1^n : f^+(n)\geq
  m\}=\emptyset$ and $\La_{\gs_f}\cap \{1^n0^m1: -f^-(n)\geq m
  \}=\emptyset$. $\gs_{f^-,f^+}$ is a transitive $F$-invariant
  subshift and, according to its definition, one has $\{h\in\Fon: f^-
  \precsim h \precsim f^+\}\subset\alk(\gs,F)$. The intuition is that,
  even if blocks of $1$ disappear only at unit speed, they are spaced
  enough in $\gs_{f^-,f^+}$ so that no curve $h$ with $f^-\precsim
  h\precsim f^+$ travel fast enough to cross a block of $0$ before the
  neighboring block of $1$ has completely disappeared.
\end{ex}

\subsection{Properties of $\aalk(\gs,F)$}

In this section, we show that the set of curves along which a CA is
equicontinuous is very constrained. The first proposition shows that
the existence of two non-equivalent such curves implies nilpotency.

\begin{prop}\label{prop:equicontdir} Let $(\az,F)$ be a CA and $\gs\subset\az$ be a
  weakly-specifed subshift. If there exist $h_1,h_2\in\aalk(\gs,F)$
  such that $h_1\nsim h_2$ then $(\az,F)$ is $\gs$-nilpotent, so
  $\aalk(\gs,F)=\Fon$.
\end{prop}
\begin{proof} 
  Because $\gs$ is weakly specified, there exists a $\sigma$-periodic
  configuration $z\in\gs$. The orbit $\{F^n(z)\}_{n\in\N}$ of $z$ is
  finite and contains only $\sigma$-periodic configurations. Let us
  consider $\gs'$ the subshift generated by this orbit. It is finite
  and therefore of finite type of some order $q$. From the definition
  of weak specificity, we also have $N\in\N$ such that for any
  configuration $x\in\gs$, there exists a word $w$ of length $n\leq N$
  such that the configuration $x_{\rrbracket -\infty, 0\rrbracket}wz_{\llbracket 0, +\infty\llbracket}$ is in
  $\gs$.
    
  We will now show that there exists $t_0\in\N$ such that for any
  configuration $x\in\gs$, $F^{t_0}(x)\in\gs'$.

  The $(\N,\gs)$-equicontinuity of $(\az,F)$ along $h_1$ and $h_2$
  implies that there exist $k, l\in \Z$, $k\leq l$, such that for all
  $x, x'\in\gs$, if $x_{\llbracket k,l\rrbracket}=x'_{\llbracket k,l\rrbracket}$ then for all $t\in\N$
  \[
  \begin{array}{rcl}
    F^t(x)_{\llbracket h_1(t), h_1(t)+q\rrbracket} & = & F^t(x')_{\llbracket h_1(t), h_1(t)+q\rrbracket}\\
    F^t(x)_{\llbracket h_2(t), h_2(t)+q\rrbracket} & = & F^t(x')_{\llbracket h_2(t), h_2(t)+q\rrbracket}
  \end{array}
  \]

  Since $h_1\nsim h_2$, there exists $t_0$ such that
  $|h_1(t_0)-h_2(t_0)| > (l-k+N)$. We will assume that
  $h_1(t_0)>h_2(t_0)$. For any configuration $x\in \gs$, by
  equicontinuity along $h_1$, $F^{t_0}(x)_{\llbracket 0,q\rrbracket}$ only depends on
  $x_{\llbracket k-h_1(t_0), l-h_1(t_0)\rrbracket}$ (not the rest of the configuration
  $x$), but by equicontinuity along $h_2$, $F^{t_0}(x)_{\llbracket 0,q\rrbracket}$ only
  depends on $x_{\llbracket k-h_2(t_0), l-h_2(t_0)\rrbracket}$.

  Because $\gs$ is weakly specified, for any configuration $x\in\gs$
  there exists a configuration $y\in\gs$ and $n\leq N$ such that (see
  Figure \ref{fig:equicontdir})
  \[ 
  \begin{array}{rcl}
    y_{\rrbracket -\infty, l-h_1(t_0)\rrbracket} & = & x_{\rrbracket-\infty, l-h_1(t_0)\rrbracket}\\
    y_{\llbracket l-h_1(t_0)+n, +\infty\llbracket} & = & z_{\llbracket 0, +\infty \llbracket}
  \end{array}
  \]

  \begin{figure}[htbp]
    \centering
    \includegraphics{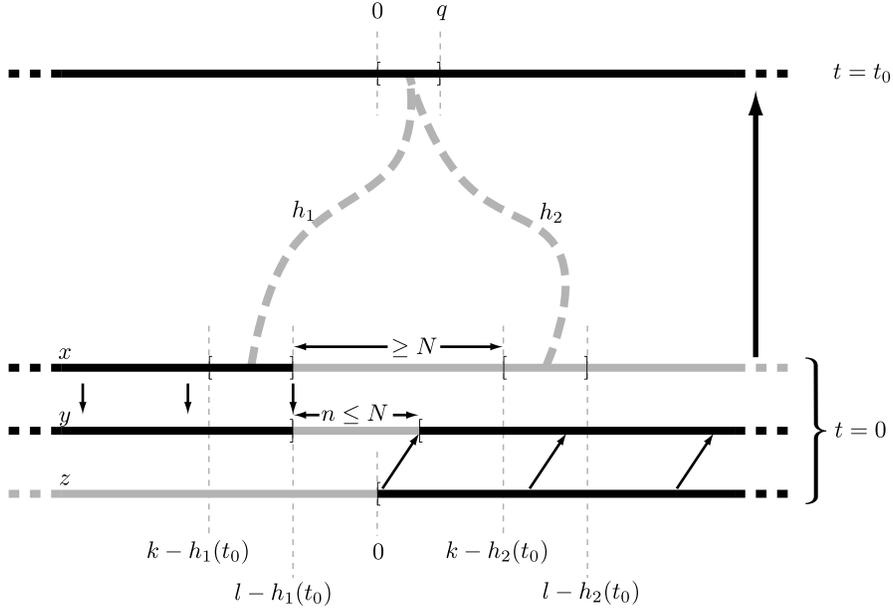}
    \caption{Construction of the configuration $y$ from a configuration $x$ and the
      periodic configuration $z$.}
\label{fig:equicontdir}
  \end{figure}
  Moreover, $\llbracket k-h_1(t_0), l-h_1(t_0)\rrbracket \subseteq\  \rrbracket-\infty,l-h_1(t_0)\rrbracket$
  and $\llbracket k-h_2(t_0), l-h_2(t_0)\rrbracket \subseteq \llbracket l-h_1(t_0)+n, +\infty\llbracket$,
  meaning that
  \[
  F^{t_0}(x)_{\llbracket 0,q\rrbracket} = F^{t_0}(y)_{\llbracket 0,q\rrbracket} = F^{t_0}(\sigma^m(z))_{\llbracket 0,q\rrbracket}
  \]
  where $m = -l+h_1(t_0)-n$.

  This shows that the factor $F^{t_0}(x)_{\llbracket 0, q\rrbracket}$ is in $\La_{\gs'}$
  (as a factor in the evolution of $\sigma^m(z)$). Because $F$
  commutes with the shift, we have shown that all factors of size $q$
  that appear after $t_0$ steps in the evolution of any configuration
  are in $\La_{\gs'}$ and since $q$ is the order of $\gs'$, it means
  that for all configuration $x$, $F^{t_0}(x)\in \gs'$. Because $\gs'$
  is finite, the CA is nilpotent.
\end{proof}

The next proposition shows that in the case of a unique curve of
equicontinuity (up to $\sim$), this curve is in fact equivalent to a
rational slope.

\begin{prop}\label{prop:equicontrational} Let $(\az,F)$ be a CA and $\gs\subset\az$ a
  subshift. If there exists $h\in\Fon$ such that $\aalk(\gs,F)=\{h'\in\Fon :
  h'\sim h\}$, then there exists $\alpha\in\Q$ such that $h\sim h_{\alpha}$.
\end{prop}
\begin{proof} Let $(\az,F)$ be a non-nilpotent CA. By definition of
  $(\N,\gs)$-equicontinuity \along{} $h$, there exist $k,l\in\Z$,
  $k\leq l$, such that for all $x, x'\in\gs$, if
  $x_{\llbracket k,l\rrbracket}=x'_{\llbracket k,l\rrbracket}$ then for all $t\in\N$ one has:
  \[
  F^t(x)_{h(t)}=F^t(x')_{h(t)}
  \]

  Thus the sequence $(F^t(x)_{h(t)})_{t\in\N}$ is uniquely determined
  by the knowledge of $x_{\llbracket k,l\rrbracket}$. For all $t\in\N$, consider the
  function
  \[
  \begin{array}{rrcl}
    f_t: & \La_{\gs}(\llbracket k,l\rrbracket) & \longrightarrow & \A \\ 
    & w & \longmapsto & F^t(x)_{h(t)} \textrm{ where } x\in [w]_{\llbracket k,l\rrbracket}\cap\gs
  \end{array}
  \]

  Because there are finitely many functions from $\La_{\gs}(\llbracket k,l\rrbracket)$ to
  $\A$, there exist $t_1, t_2\in \N$ such that $t_1<t_2$ and
  $f_{t_1}=f_{t_2}$.

  For any configuration $x\in\gs$, and any cell $c\in \Z$,
  \[
  F^{t_1}(x)_{h(t_1)+c}= f_{t_1}(x\llbracket k+c, l+c\rrbracket) = f_{t_2}(x\llbracket k+c, l+c\rrbracket) = F^{t_2}(x)_{h(t_2)+c}
  \]

  We therefore have $F^{t_1}(x) = \s^{h(t_2)-h(t_1)} \circ
  F^{t_2}(x)$ for all possible configurations $x\in\gs$. With $\alpha
  = \frac{h(t_2)-h(t_1)}{t_2-t_1}$, $h_\alpha$ is a direction of
  equicontinuity of $(\az,F)$.
\end{proof}

\subsection{Properties of $\bl(\gs,F)$}

% \begin{theorem}\label{caracteristique_bl} Let $(\az,F)$ be a CA of
% neighborhood $\U=[r,s]$ and $\gs\subset\az$ an infinite transitive
% subshift.
%   \begin{enumerate}
%   \item[$\bullet$]If $\bln^{r}(\gs,F)\ne\emptyset$ then there exists
% $\alpha'\geq -s$ such that $\bln^{r}(\gs,F)=\{ h\in\Fon :
% h_{\alpha'}\prec h\}$.

%   \item[$\bullet$]If $\bln^{l}(\gs,F)\ne\emptyset$ then there exists
% $\alpha''\leq -r$ such that $\bln^{l}(\gs,F)=\{h\in\Fon : h\prec
% h_{\alpha''}\}$.

%   \item[$\bullet$]If $\bln(\gs,F)\ne\emptyset$ then there exist
% $\alpha',\alpha''\in\R$ with $-s\leq \alpha'\leq\alpha''\leq -r$ such
% that $\bln(\gs,F)=\{h\in\Fon : h_{\alpha'} \prec h\prec
% h_{\alpha''}\}$.
%   \end{enumerate}
% \end{theorem}
% \begin{proof} \TODO{Pour l'instant ce n'est encore qu'une
% conjecture.}
% \end{proof}

The next proposition shows the link between expansivity and
equicontinuous properties.
\begin{prop}\label{prop:expequ} Assume $\K=\N$ or $\Z$. Let $(\az,F)$ be a CA, $\gs$ be
an infinite subshift. One has:
  $$\left(\blk^{r}(\gs,F)\cup\blk^{l}(\gs,F)\right)\cap\alk(\gs,F)=\emptyset.$$ 
  In particular, if $\blk(\gs,F)\ne\emptyset$ then
$\alk(\gs,F)=\emptyset$.
\end{prop}

\begin{proof} Let $(\az,F)$ be $(\K,\gs)$-right expansive \along{} $h$
  with constant of expansivity $\e$. One
  has: $$\tube^{h}_{\gs}(x,\e,\K)\subset\{y\in\gs : y_i=x_i\ \forall
  i\geq 0\}.$$ Then the interior of $\tube^{h}_{\gs}(x,\e,\K)$ is
  empty. Thus $Eq^{h}_{\K}(\gs,F)=\emptyset$.

  Analogously, one proves
$\blk^{l}(\gs,F)\cap\alk(\gs,F)=\emptyset$. In the case
$\blk(\gs,F)\ne\emptyset$, one has
$\blk^{r}(\gs,F)\cup\blk^{l}(\gs,F)=\Fon$, so $\alk(\gs,F)=\emptyset$.
\end{proof}

\subsection{A dynamical classification along a curve}

\begin{theorem}\label{CuveClassification} Let $(\az,F)$ be a CA of
  neighborhood $\U=\llbracket r,s\rrbracket$. Let $\gs\subset\az$ be a weakly-specified
  subshift. Exactly one of the following cases hold:
  \begin{enumerate}
  \item[\textbf{C1.}] $\aaln(\gs,F)=\aln(\gs,F)=\Fon$. In this case
$(\az,F)$ is $\gs$-nilpotent, moreover
$\bln^{r}(\gs,F)=\bln^{l}(\gs,F)=\emptyset$.

  \item[\textbf{C2.}] There exists $\alpha\in[-s,-r]\cap\Q$ such that
$\aaln(\gs,F)=\aln(\gs,F)=\{h: h\sim h_{\alpha}\}$. In this case there
exist $m,p\in\N$ such that the sequence $(F^n\circ\s^{\lfloor \alpha
n\rfloor})_{n\in\N}$ is ultimately periodic of preperiod $m$ and
period $p$. Moreover, $\bln^{l}(F^m(\gs),F)=]-\infty,\alpha[$ and
$\bl^{r}(F^m(\gs),F)=]\alpha,+\infty[$.

  \item[\textbf{C3.}] There exist $h',h''\in\Fon$, $h'\prec h''$,
$h''\precsim h_{-r}$ and $h_{-s}\precsim h''$ such that $\{h :h'\prec
h\prec h''\}\subset\aln(\gs,F)\subset \{h : h'\precsim h \precsim
h''\}$. In this case
$\aaln(\gs,F)=\bln^{r}(\gs,F)=\bln^{l}(\gs,F)=\emptyset$.

  \item[\textbf{C4.}] There exists $h'\in\Fon$, $h_{-s}\precsim
h'\precsim h_{-r}$, such that $\aln(\gs,F)=\{h : h\sim h'\}$ and
$\aaln(\gs,F)=\emptyset$. In this case $\bln^{r}(\gs,F)$ and
$\bln^{l}(\gs,F)$ can be empty or not, but $\bln(\gs,F)=\emptyset$.

\item[\textbf{C5.}] $\aln(\gs,F)=\aalk(\gs,F)=\emptyset$ but
  $\bln(\gs,F)\not=\emptyset$.
    
\item[\textbf{C6.}] $\aln(\gs,F)=\aaln(\gs,F)=\bln(\gs,F)=\emptyset$ but
  $\bln^{r}(\gs,F)$ and $\bln^{l}(\gs,F)$ can be empty or not.
  \end{enumerate}
\end{theorem}
\begin{proof}
  First, by proposition~\ref{prop:equicontdir} and considering the possible
  values of $\aaln(\gs,F)$, we get a partition into: $C1$, $C2$, and $C'=C3\cup
  C4\cup C5\cup C6$. The additional property in class $C2$ is obtained by
  proposition~\ref{prop:equicontrational}.

  Then, inside $C'$, the partition is obtained by discussing on $\aln(\gs,F)$
  (proposition~\ref{convexe_al}) and $\bln(\gs,F)$, non-emptyness of both being
  excluded by proposition~\ref{prop:expequ}.
\end{proof}

\section{Equicontinuous dynamics: non-trivial constructions}
\label{sec:examples}

This section aims at showing through non-trivial examples that the
generalization of directional dynamics to arbitrary curve is pertinent.

\subsection{Parabolas}
\label{subsec:par}

\newcommand{\stw}{\vbox to 7pt{\hbox{\includegraphics{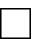}}}}
\newcommand{\sta}{\vbox to 7pt{\hbox{\includegraphics{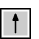}}}}
\newcommand{\str}{\vbox to 7pt{\hbox{\includegraphics{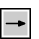}}}}
\newcommand{\stl}{\vbox to 7pt{\hbox{\includegraphics{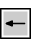}}}}
\newcommand{\stb}{\vbox to 7pt{\hbox{\includegraphics{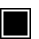}}}}

Let us define the function
\[ p: \left\{
  \begin{array}{rcl} \N & \rightarrow & \Z \\ x & \mapsto &
    \left\lfloor \frac{\sqrt{1+4(x+1)}-1}{2} \right\rfloor
  \end{array} \right.
\]
whose inverse is
\[ p^{-1}: \left\{
  \begin{array}{rcl} \N & \rightarrow & \Z \\ x & \mapsto & x(x+1)-1
  \end{array} \right.
\]

This whole subsection will be devoted to the proof and discussion of
the following result~:

\begin{prop}\label{prop:par} 
  There exists a cellular automaton $(\apz,\Fp)$ such that
  \[
  \aln(\apz,\Fp)= \{h\in\Fon : p\precsim h \precsim
  \operatorname{id}\}
  \] 
  where $\operatorname{id}$ denotes the identity function $n\mapsto
  n$.
\end{prop}

\begin{proof}

Let us describe such an automaton. We will work on the standard
neighborhood $\Up=\{-1,0,1\}$ and use the set of 5 states $\Ap=\{\stw,
\str, \stl, \sta, \stb\}$. The behavior of the automaton will be
described in terms of \emph{signals}~: a cell in state $\stw$ should
be seen as an \emph{empty cell} with no signal, whereas all other
states represent a given signal on the cell. There can only be one
signal at a time on a given cell.

Signals move through the configuration. A signal can move to the left,
to the right or stay on the cell it is (in which case we will say that
the signal moves up because it makes a vertical line on the space-time
diagram) as shown on Figure~\ref{fig:par:signals}. A signal can also
duplicate itself by going in two directions at a time (last case in
figure \ref{fig:par:signals}).

\begin{figure}[htbp] \centering
  \includegraphics{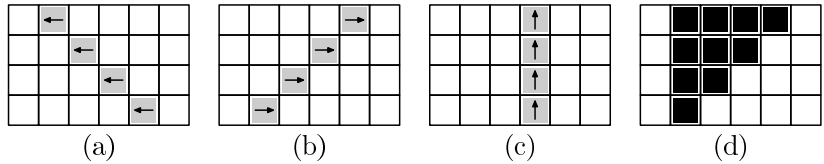}
  \caption{Signals moving left (a), right (b), up (c) and both up and
right (d).}
  \label{fig:par:signals}
\end{figure}

We will now describe how each signal moves when it is alone
(surrounded by empty cells) and how to deal with collisions, when two
or more signals move towards the same cell (figure
\ref{fig:par:std_ex} provides a space-time diagram that illustrates
most of these rules)~:

\begin{itemize}
\item the $\stb$ signal moves up and right. It has priority over all
signals except the $\str$ one (signals with lesser priority disappear
when a conflict arises)~;
\item the $\sta$ signal moves up. It has priority over all signals
except the aforementioned $\stb$ signal~;
\item the $\stl$ signal moves left until it reaches a $\sta$ signal,
at which point it becomes a $\str$ signal (it turns around) instead of
colliding into it~;
\item finally, the $\str$ signal moves right until it reaches a $\stb$
signal in which case it moves over it but turns into a $\stl$ signal
(and therefore from there it moves \emph{away} from the $\stb$). Not
only the $\str$ signal cannot go through a $\sta$ signal, as a
consequence of what has been stated earlier, but it cannot cross a
$\stl$ signal either (even if there was no real collision because the
two could switch places)~: it is erased by the $\stl$ signal moving in
the opposite direction.
\end{itemize}

\begin{figure}[htbp] \centering
  \includegraphics{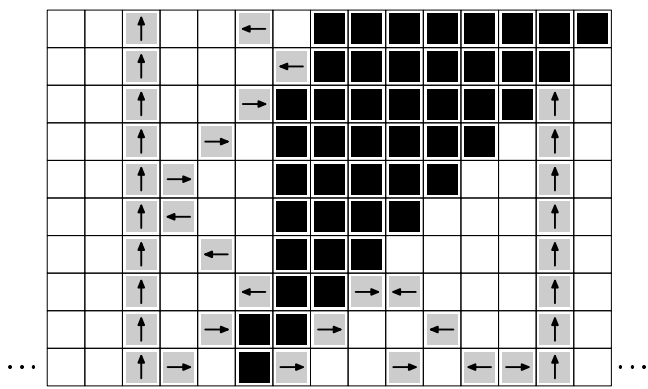}
  \caption{A space-time diagram of the described automaton on an
example starting configuration.}
  \label{fig:par:std_ex}
\end{figure}

The general behavior of the automaton can be described informally as
follows:
\begin{itemize}
\item $\stb$ states form connex segments that expand towards the right
and can be reduced from the left by $\str$ signals~;
\item $\sta$ signals create ``vertical axes'' on the space-time
  diagram~;
\item $\str$ and $\stl$ signals bounce back and forth from a vertical
$\sta$ axis (on the left) to a $\stb$ segment (on the right). The
$\sta$ border does not move but the $\stb$ on the right side is pushed
to the right at each bounce~;
\item $\stb$ signals can erase $\sta$ signals and by doing so
``invade'' a portion in which bouncing signals evolve. $\stb$ segments
can merge when the right border of one reaches the left border of
another.
\end{itemize}

We will now show that $\cone{\Fp}{[\stb]_0}$, the set of consequences
of the single-letter word $w=\stb$ according to this automaton, is
exactly the set of sites
\[\paraset=\{\site ct \st t\geq 0, c\in
  \llbracket p(t), t\rrbracket\}
\]

\fact{ $\paraset$ is exactly the set of sites in state
$\stb$ in the space-time diagram starting from the initial finite
configuration corresponding to the word $\sta\str\stb$ (all other
cells are in state $\stw$), as illustrated by
Figure~\ref{fig:par:parabola}.}\label{fact:par:cc}

\bfactprf It is clear from the behavior of the automaton that it takes
$2(n+1)$ steps for the left border of the $\stb$ segment to move from
cell $n$ to cell $(n+1)$. Conveniently enough, $p$ has the property
that $p^{-1}(n+1)=p^{-1}(n)+2(n+1)$.

\efactprf

\begin{figure}[htbp] \centering
  \includegraphics{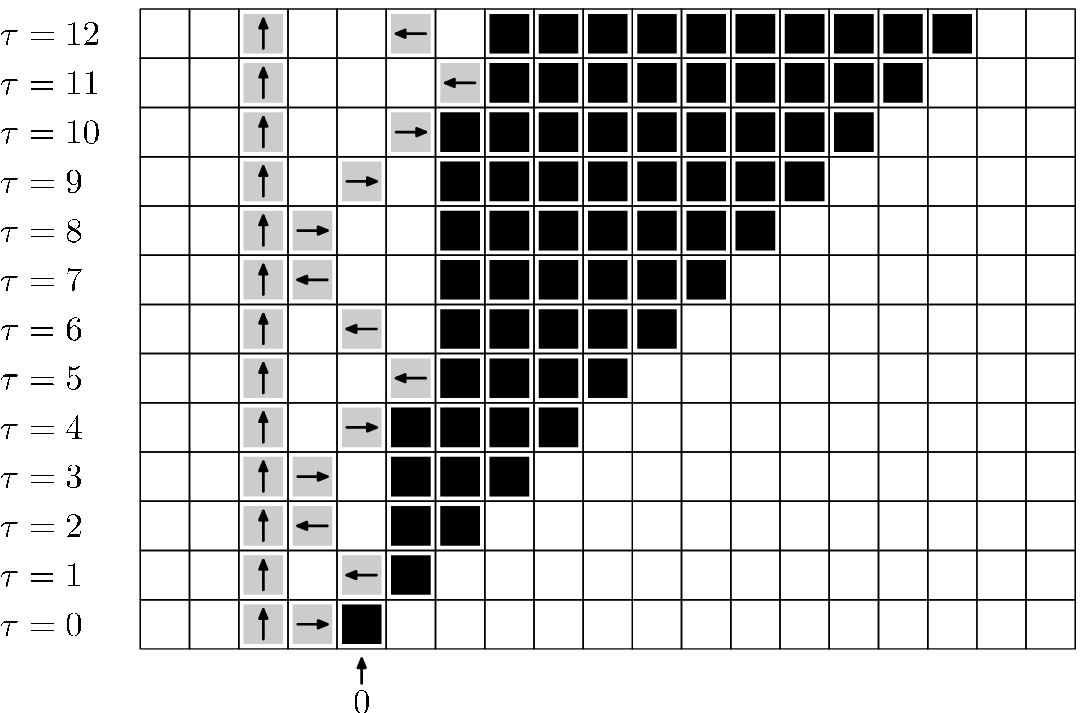}
  \caption{The set $\paraset$ is exactly the set of sites in
    state \protect\includegraphics[scale=0.75]{figures/dpst-5}.}\label{fig:par:parabola}
\end{figure}

From Fact~\ref{fact:par:cc} we show that all sites that are not in
$\paraset$ cannot be in the consequences of $w$ because if we start
from the uniformly $\stb$ configuration (which is an extension of $w$)
all states in the diagram are $\stb$ and hence these sites have
different states depending on the extension of $w$ used as starting configuration.

Let us prove that conversely, for whatever starting configuration that contains $w$ at the origin,
all sites in $\paraset$ are in state $\stb$.

\fact{If there exists a starting configuration $\mathcal{C}$ containing $w$ at the origin such
  that one of the sites $\site ct\in\paraset$ is in a state other than
  $\stb$, then there exists a \emph{finite} such starting configuration
  (one for which all cells but a finite number are in state $\stw$)
  for which the site $\site ct$ is in a state other than
  $\stb$.}\label{fact:par:finite} \bfactprf The state in the site
$\site ct$ only depends on the initial states of the cells in
$\llbracket c-t, c+t\rrbracket$. The finite configuration that
coincides with $\mathcal C$ on these cells and contains only $\stw$
states on all other cells has the announced property.  \efactprf

The $\stb$ signal tends to propagate towards the top and the
right. Since only the $\str$ signal has priority over the $\stb$ one
and because the former moves to the right, it cannot collide with the
latter from the right side, which means that nothing can hinder the
evolution of the $\stb$ signal to the right.

From now on, we will say that a connex segment of cells in state
$\stb$ (that we will simply call a $\stb$ segment) is \emph{pushed}
whenever a $\str$ signal bounces on its left border, and by doing so
erases the leftmost $\stb$ state of the segment.

\fact{ Starting from an initial finite
configuration, the time interval between two consecutive ``pushes'' of
the leftmost $\stb$ state (by a $\str$ signal) is exactly double the
distance between it and the first $\sta$ state to its left, if any.} \label{fact:par:leftmost}

\bfactprf All $\sta$ signals to the left of the leftmost $\stb$
segment are preserved. When a $\str$ signal pushes the $\stb$ state,
it generates a $\stl$ signal that moves left erasing all $\str$
signals it meets on its way. Therefore nothing can reach the $\stb$
state while the $\stl$ signal is moving. When it reaches the first
$\sta$ state (if there is one) and turns into a $\str$ signal, the
configuration is as follows~:
  \[
  \includegraphics{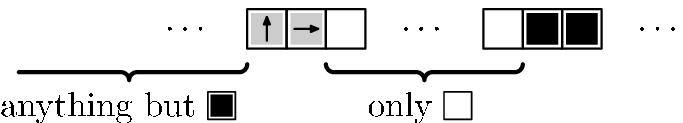}
  \] and nothing other than this newly produced $\str$ signal will
push the $\stb$ state. The time between the apparition of the $\stl$
signal after the first push until the push by the second $\str$ signal
is exactly double the distance between the $\stb$ and the $\sta$
states.

  \rmk If there are no $\sta$ states to the left of the $\stb$ then
there can be at most one push because the $\stl$ will never bounce
back and will erase all $\str$ signals before they reach the $\stb$
state.
\efactprf

\fact{If $c$ is the leftmost cell in state $\stb$ of
a finite starting configuration, then all sites in $(\paraset+c)$ are in
the state $\stb$.}\label{fact:par:left}

\bfactprf From Fact~\ref{fact:par:leftmost} we show that the
configuration corresponding to the word $\sta\str\stb$ illustrates the
fastest way to push a leftmost $\stb$ segment~: it is the
configuration where the $\sta$ signal is the closest possible to the
$\stb$ (while still having a bouncing signal between them) and for
which the first push happens at the earliest possible time.

This means, in conjunction with Fact~\ref{fact:par:cc}, that if $c$ is
the leftmost cell originally in state $\stb$ then all sites in
$(\paraset+c)$ are in state $\stb$ because its $\stb$ segment cannot
be pushed faster.  \efactprf

Fact~\ref{fact:par:left} can be extended to all cells in state $\stb$
by induction~:

\fact{If a cell $c$ is in state $\stb$ in a
finite starting configuration and that for all cells $c'<c$ initially
in state $\stb$ all sites in $(\paraset+c')$ are in state $\stb$, then all
sites in $(\paraset+c)$ are in state $\stb$.}  \label{fact:par:induction}
\bfactprf Let $c_1$ be the closest cell to the left of $c$ that is
initially in state $\stb$. Because the $\stb$ signal from $c_1$
propagates to the right at maximal speed, it can have no influence on
the behavior of the $\stb$ segment generated by $c$ before it has
actually reached it.

  This means that, until the two $\stb$ segments merge (at some time
$t_1$), the segment from $c$ behaves as if it were the leftmost one,
and therefore Fact~\ref{fact:par:left} applies and ensures that all
sites in
  \[ \{\site ct \st t<t_1, c\in\llbracket c+p(t), c+t\rrbracket\}
  \] are in state $\stb$.

  After the two segments merge, we know that all sites in
  \[ \{\site ct \st t\geq t_1, c\in\llbracket c_1+p(t),
c+t\rrbracket\}
  \] are in state $\stb$. These include all remaining sites in
$(\paraset+c)$ since $c_1<c$.
\efactprf

By Fact~\ref{fact:par:finite}, it is sufficient to show that for all
finite extensions of $w$ all sites in $\paraset$ are in the $\stb$ state.

We then proceed by induction to show that for any cell $c$ in state
$\stb$ on a finite initial configuration, all sites in $(\paraset+c)$
are in state $\stb$ (Fact~\ref{fact:par:left} is the initialization,
Fact~\ref{fact:par:induction} is the inductive step).

This concludes the proof that $\cone{\Fp}{[\stb]_0} = \paraset$ and
therefore
\[\{h\in\Fon : p\precsim h \precsim \operatorname{id}\}\subseteq \aln(\apz,\Fp)\]

For the converse inclusion, let $h\in\aln(\apz,\Fp)$ and suppose that
$w$ is a blocking word along $h$. Then the word $vwv$ is also a
blocking word along $h$ with $v=\sta\str\stb$. From the definition of
$\Fp$ we know that for any site $\site zt$ in the consequences of
$vwv$, with $t\geq |vwv|$, and any configuration $x\in[vwv]_0$ then
$\Fp^t(x)_z =\stb$ (because it is the case for the configuration
everywhere in state $\stb$ except on the finite portion where it is
$vwv$). Now, if we consider the sites in state $\stb$ generated by the
configuration everywhere in state $\stw$ except on the finite portion
where it is $vwv$, we have:
\[
\site zt \in\cone{\Fp}{[vwv]_0} \Rightarrow\ p(t)+C_1\leq z\leq \operatorname{id}(t)+C_2,
\] 
for some constants $C_1$ and $C_2$.

This completes the proof of Proposition~\ref{prop:par}.
\end{proof}

This shows that the notion of equicontinuous points following a
non-linear curve is pertinent. This was an open question of
\cite{Sablik-2008}.

An other open question of \cite{Sablik-2008} was to find a cellular
automaton such that $\daln(\az,F)$ has open bounds.

\begin{corollary}
  There exists a CA $(\az,F)$ such that $\daln(\az,F)$ has open
  bounds.
\end{corollary}
\begin{proof}
  Choosing $F=\Fp$, the result follows directly from
  Proposition~\ref{prop:par}: the set of slopes of straight lines lying
  between $p$ and $\operatorname{id}$ is exactly $]0,1]$.
\end{proof}

%%%%%%%%%%%%%%%%%%%%%%%%%%%%%%%
%%%
%%%%%%%%%%%%%%%%%%%%%%%%%%%%%%%

\subsection{Counters}
\label{subsec:counters}

\newcommand{\sti}{\vbox to 7pt{\hbox{\includegraphics{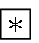}}}}

In this section we will describe a general technique that can be used
to create sets of consequences that have complex shapes.

The general idea is to build the set of consequences in a
``protected'' area in a cone of the space-time diagram (the area
between two signals moving in opposite directions) and making sure
that nothing from the outside can affect the inside of the cone.

We will illustrate the technique on a specific example that will
describe a CA for which the set of consequences of a single-letter
word is the area between the vertical axis and a parabola. The
counters construction will then also be used in section
\ref{sec:linear} to construct more complex sets of consequences.

\begin{prop}\label{prop:counters}
  There exists a cellular automaton $(\az_C,F_C)$ such that
  \[
  \aln(\az_C,F_C) = \{h\in\Fon : \monbb{0} \precsim h \precsim p\}
  \]
  where $\monbb{0}$ denotes the constant fuction $n\mapsto 0$.
\end{prop}

As in the previous section, we will show that the consequences of a
single-letter word are exactly the set
\[ 
\{\site ct \st t\geq 0, c\in \llbracket 0,p(t)\rrbracket\}
\]

\subsubsection{General Description}

The idea is to use a special state $\sti$ that can only appear in the
initial configuration (no transition rule produces this state). This
state will produce a cone in which the construction will take
place. On both sides of the cone, there will be unary counters that
count the ``age of the cone''.

The counters act as protective walls to prevent external signals from
affecting the construction. If a signal other than a counter arrives,
it is destroyed. If two counters collide, they are compared and the
youngest has priority (it erases the older one and what comes
next). Because our construction was generated by a state $\sti$ on the
initial configuration, no counter can be younger (all other counters
were already present on the initial configuration).

The only special case is when two counters of the same age collide. In
this case we can ``merge'' the two cones correctly because both
contain similar parabolas (generated from an initial special state).

\subsubsection{Constructing the Parabola inside the Cone}

Inside the cone, we will construct a parabola by a technique that
differs from the one explained in the previous section in that the
construction signals are on the outer side.

The construction is illustrated by Figure \ref{fig:parabola_right}. We
see that we build two inter-dependent parabolas by having a signal
bounce from one to the other. Whenever the signal reaches the left
parabola, it drags it one cell to the right, and when it reaches the
righ parabola it pushes it two cells to the right.

\begin{figure}[htbp] \centering
  \includegraphics{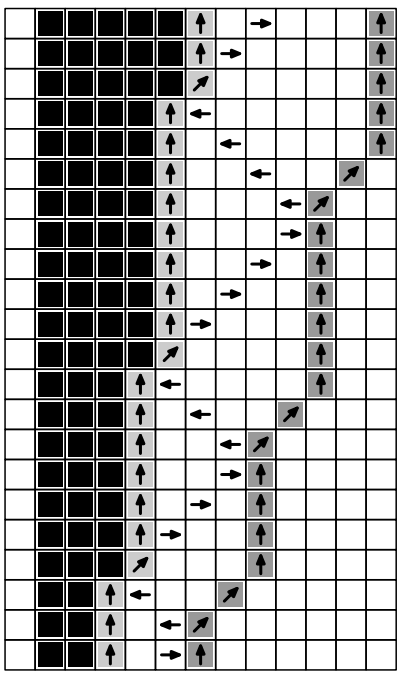}
  \caption{Construction of the parabola inside the cone.}
  \label{fig:parabola_right}
\end{figure}

It is easy to check that the left parabola advances as the one from
the previous section : when it has advanced for the $n$-th time, the
right parabola is at distance $n$ and therefore the signal will come
back after $2n$ time steps.

The $\stb$ state (that will be the set of consequences, as in the
previous section) moves up and right, but it is stopped by the states
of the left parabola.

\subsubsection{The Younger, the Better}

The $\sti$ state produces 4 distinct signals. Two of them move towards
the left at speed $1/4$ and $1/5$ respectively. The other two move
symmetrically to the right at speed $1/4$ and $1/5$.

Each couple of signals (moving in the same direction) can be seen as a
unary counter where the value is encoded in the distance between the
two. As time goes the signals move apart.

Note that signals moving in the same direction (a fast one and a slow
one) are not allowed to cross. If such a collision happens, the faster
signal is erased. A collision cannot happen between signals generated
from a single $\sti$ state but could happen with signals that were
already present on the initial configuration. Collisions between
counters moving in opposite directions will be explained later as
their careful handling is the key to our construction.

Because the $\sti$ state cannot appear elsewhere than the initial
configuration and counter signals can only be generated by the $\sti$
state (or be already present on the initial configuration), a counter
generated by a $\sti$ state is at all times the smaller possible one:
no two counter signals can be closer than those that were generated
together. Using this property, we can encapsulate our construction
between the smallest possible counters. We will therefore be able to
protect it from external perturbations: if something that is not
encapsulated between counters collides with a counter, it is
erased. And when two counters collide we will give priority to the
youngest one.

\subsubsection{Dealing with collisions}

Collisions of signals are handled in the following way:
\begin{itemize}
\item nothing other than an \emph{outer} signal can go through another
\emph{outer} signal (in particular, no ``naked information'' not
contained between counters);
\item when two \emph{outer} signals collide they move through each
other and comparison signals are generated as illustrated by Figure
\ref{fig:speeds};
\item on each side, a signal moves at maximal speed towards the
\emph{inner} border of the counter, bounces on it ($C$ and $C'$) and
goes back to the point of collision ($D$);
\item The first signal to come back is the one from the youngest
counter and it then moves back to the \emph{outer} side of the oldest
counter ($E$) and deletes it;
\item the comparison signal from the older counter that arrives
afterwards ($D'$) is deleted and will not delete the younger counter's
\emph{outer} border;
\item all of the comparison signals delete all information that they
encounter other than the two types of borders of counters.
\end{itemize}

\begin{figure}[htbp] \centering
  \includegraphics[height=6cm]{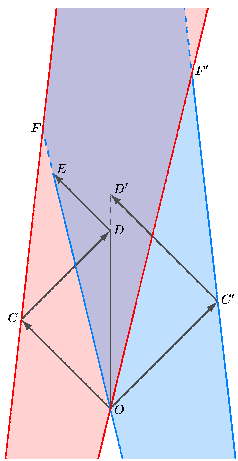}
  \caption{The bouncing signal must arrive (point $E$) before the
older counter moves through the younger one (point $F$).}
  \label{fig:speeds}
\end{figure}

\subsubsubsection{Counter Speeds} It is important to ensure that the
older counter's \emph{outer} border is deleted before it crosses the
younger's \emph{inner} border. This depends on the speeds $s_o$ and
$s_i$ of the \emph{outer} and \emph{inner} borders. It is true
whenever $s_o \leq \frac{1-s_i}{s_i+3}$. If the maximal speed is $1$
(neighborhood of radius $1$), it can only be satisfied if
\[ s_i<\sqrt{5}-2\simeq 0.2360
\] This means that with a neighborhood of radius 1 the \emph{inner}
border of the counter cannot move at a speed greater than $(\sqrt
5-2)$. Any rational value lower than this is acceptable. For
simplicity reasons we will consider $1/5$ (and the corresponding $1/4$
for the \emph{outer} border of the counter). If we use a neighborhood
of radius $k$, the counter speeds can be increased to $k/5$ and $k/4$.

\subsubsubsection{Exact Location}

Note that a precise comparison of the counters is a bit more complex
than what has just been described. Because we are working on a
discrete space, a signal moving at a speed less than maximal does not
actually move at each step. Instead it stays on one cell for a few
steps before advancing, but this requires multiple states.

In such a case, the cell on which the signal is is not the only
significant information. We also need to consider the current state of
the signal: for a signal moving at speed $1/n$, each of the $n$ states
represents an advancement of $1/n$, meaning that if a signal is
located on a cell $c$, depending on the current state we would
consider it to be exactly at the position $c$, or $(c+1/n)$, or
$(c+2/n)$, etc. By doing so we can have signals at rational
non-integer positions, and hence consider that the signal really moves
at each step.

When comparing counters, we will therefore have to remember both
states of the faster signals that collide (this information is carried
by the vertical signal) and the exact state in which the slower signal
was when the maximal-speed signal bounced on it. That way we are able
to precisely compare two counters: equality occurs only when both
counters are exactly synchronized.

\subsubsubsection{The \emph{Almost} Impregnable Fortress}

Let us now consider a cone that was produced from the $\sti$ state on
the initial configuration. As it was said earlier, no counter can be
younger than the ones on each side of this cone. There might be other
counters of exactly the same age, but then these were also produced
from a $\sti$ state and we will consider this case later and show that
it is not a problem for our construction.

Nothing can enter this cone if it is not preceded by an \emph{outer}
border of a counter. If an opposite \emph{outer} border collides with
our considered cone, comparison signals are generated. Because
comparison signals erase all information but the counter borders, we
know that the comparison will be performed correctly and we do not
need to worry about interfering states. Since the borders of the cone
are the youngest possible signals, the comparison will make them
survive and the other counter will be deleted.

Note that two consecutive opposite \emph{outer} borders, without any
\emph{inner} border in between, are not a problem. The comparison is
performed in the same way. Because the comparison signals cannot
distinguish between two collision points (the vertical signal from $O$
to $D$ in Figure \ref{fig:speeds}) they will bounce on the first they
encounter. This means that if two consecutive \emph{outer} borders
collide with our cone, the comparisons will be made ``incorrectly''
but this error will favor the well formed counter (the one that has an
\emph{outer} and an \emph{inner} border) so it is not a problem to us.

\subsubsubsection{Evil Twins}

The last case we have to consider now is that of a collision between
two counters of exactly the same age. Because the only counter that
matters to us is the one produced from the $\sti$ state (the one that
will construct the parabola), the case we have to consider is the one
where two cones produced from a $\sti$ state on the initial
configuration collide. These two cones contain similar parabolas in
their interior.

According to the rules that were described earlier, both colliding
counters are deleted. This means that the right side of the leftmost
cone and the left part of the rightmost cone are now ``unprotected''
and facing each other. From there, the construction of the two
parabolas will merge, as illustrated in Figure \ref{fig:merge}.

\begin{figure}[htbp] \centering
  \includegraphics{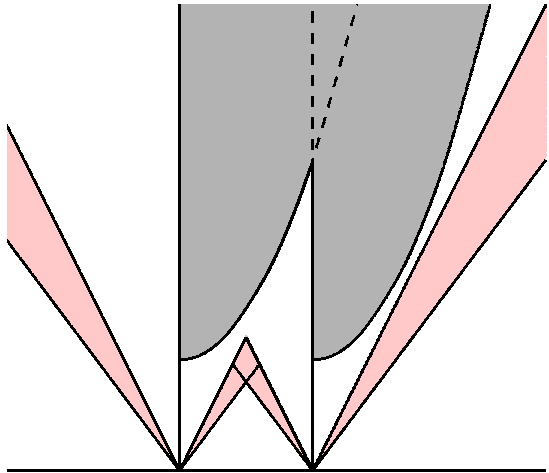}
  \caption{Two ``twin'' parabolas merging.}
  \label{fig:merge}
\end{figure} The key point here is the fact that the two merged
parabolas were ``parallel'' meaning that the one on the left would
never have passed beyond the one on the right. Because of this, for as
long as the rightmost parabola is constructed correctly, all states
that should have been $\stb$ because of the construction of the left
parabola will be $\stb$.

\subsubsection{The Transparency Trick}
\label{subsec:transparency}

We have shown that all sites that would be $\stb$ in the evolution of
a ``lonely'' $\sti$ state (all other initial cells being in the blank
$\stw$ state) will stay in the $\stb$ state no matter what is on the
other cells of the initial configuration. Now we have to show that no
other sites than these are in the consequences of the $\sti$ state.

To do so, we will use a very simple trick. We consider our automaton
as it was described so far, and add a binary layer with states $0$ and
$1$. This layer is free and independent of the main layer, except if
the state on the main layer is $\stb$ in which case the binary layer
state can only be $0$. More precisely, this layer is kept unchanged
except when entering into state $\stb$ on the main layer in which case
it becomes $0$.  Let's call $\Fc$ the final CA obtained.

Now, if we consider configurations where the main layer is made of a
single $\sti$ state on a blank configuration, we are guaranteed that
$\stb$ states cannot disappear once present in a cell. Hence, at any
time, any cell not holding the state $\stb$ contains the value from
the initial configuration in this additional layer.  This means that,
for such initial configurations, all sites can be changed except those
holding state $\stb$. This means that the consequences of the
single-letter word $(\sti,0)$ (this is a single state in the
two-layered automaton) are exactly
\[
\cone{\Fc}{[(\sti,0)]_0} = \{\site ct \st t\geq 0, c\in \llbracket
0,p(t)\rrbracket\}
\]

This concludes the proof of Proposition \ref{prop:counters}.\\

The counters construction used to protect the evolution of the
parabola gives $\Fc$ a special property: all equicontinuous points are
Garden-of-Eden configurations. To our knowledge, this is the first
constructed CA with this property.

\begin{corollary}
  \label{cor:gardenedenwall}
  There exists a CA $(\az,F)$ having (classical) equicontinuous points, but
  all being Garden-of-Eden configurations (\textit{i.e.}
  configurations without a predecessor).
\end{corollary}
\begin{proof}
  We choose $F=\Fc$. The corollary follows from the fact that any
  blocking word must contain the state $\sti$. To see this suppose
  that some blocking word $w$ do not contain the state $\sti$. Then,
  whatever $w$ is, it cannot produce a protective cone as ``young'' as
  the one generated by $\sti$. Therefore, if $c\in[w]_0$ contains the
  state $\sti$ somewhere, the outer signal generated by $\sti$ will
  reach the central cell at some time depending on the position of the
  first occurrence of $\sti$ in $c$, and, after some additional time,
  the central cell will become $\stb$ and stay in this state
  forever. This way, we can choose two configurations $c,c'\in[w]_0$
  such that the sequence of states taken by the central cell are
  different: this is in contradiction with $w$ being a blocking word.
\end{proof}

By combining the two constructions from Propositions \ref{prop:par}
and \ref{prop:counters} we can obtain a CA with only one direction (up
to $\sim$) along which the CA has equicontinuity points and such that
this direction is not linear. This is another example where the
generalization of directional dynamics to arbitrary curves is
meaningful.

\begin{corollary}
  There exists a CA $(\az,F)$ with $\aln(\az,F)= \{h : h\sim h_0\}$
  where $h_0$ is not linear.
\end{corollary}
\begin{proof}
  It is straightforward that for any pair of CA $G,H$ we have:
  \[\aln(\az_G\times\az_H,G\times H) = \aln(\az_G,G)\cap\aln(\az_H,H)\]
  Let $F=\Fp\times\Fc$. By propositions \ref{prop:par} and
  \ref{prop:counters}, $\aln(\apz\times\az_C,F)= \{h : h\sim
  p\}$.
\end{proof}

%%%%%%%%%%%%%%%%%%%%%%%%%%%%%%%
%%%
%%%%%%%%%%%%%%%%%%%%%%%%%%%%%%%

\section{Equicontinuous dynamics along linear directions}
%\section{Bounds for $\dalk(\gs,F)=\alk(\gs,F)\cap\Dir$}
\label{sec:linear}

By a result from \cite{Sablik-2008}, we know that the set of slopes of
linear directions along which a CA has equicontinuous points is an
interval of real numbers. In this section we are going to study
precisely the possible bounds of such intervals.

\subsection{Countably enumerable numbers}

\begin{definition}
  A real number $\alpha$ is {\em countably enumerable (ce)} if there
  exists a computable sequence of rationals converging to $\alpha$.
\end{definition}

The previous definition can be further refined as follows:
\begin{definition}
  A real number $\alpha$ is {\em left (resp. right) countably enumerable (lce)
    (resp. rce)} if there exists an increasing (resp. decreasing) computable
  sequence of rationals converging to $\alpha$.
\end{definition}

\begin{remark}
  A real number that is both lce and rce is computable.
\end{remark}

See \cite{Spies-Weihrauch-Zheng-2000} for more details.\\

In the following we first prove that the bounds of the interval of
slopes of linear directions along which a CA has equicontinuous points
are computably enumerable real numbers. We will then give a generic
method to construct a CA having arbitrary computably enumerable
numbers as bounds for the slopes along which it has equicontinuity
points.

\subsection{Linear directions in the consequences of a word}

We consider in this section {\em computable} subshifts only, that is
subshifts for which we can decide whether or not a given word belongs
to it (in particular, all finite type and sofic subshifts are decidable).  For a word $u\in\gs$ and
a CA $(\az,F)$, we note
\[I_u=\{\alpha\in \R\st u\ is\ a\ blocking\ word\ of\ slope\
h_{\alpha}\ for\ F\}.\] $I_u$ is an interval and we note it
$|a_u,b_u|$. The bounds can either be open or closed.  We will show
that this bounds are lce (resp. rce). 

The proof uses the notion of \emph{blocking word along $h$ during a time $T$},
which intuitively is a word that doesn't let information go through
its consequences before time $(T+1)$. The definition can be adapted
from definition~\ref{def:blockingword} by replacing $\K$ by $\llbracket 0,T\rrbracket$, formally:
\[\cone{F}{\gs\cap[u]_p}\supset\left\{ \site{m}{n} \in\Z\times\llbracket 0,T\rrbracket: h(n)\leq m < h(n)+e \right\}.\]

Notice that if some word is blocking of slope $h$ during arbitrary
long time, it is clearly $\N$-blocking of slope $h$ too.

If time $T$ is fixed, the set of slopes for which a word $u$ is
blocking during time $T$ is a convex set. This is formalized by the
following Lemma which is a weakened version of
Proposition~\ref{convexe_al}.

\begin{lemma}
\label{lem:esayconvex}
  Let $F$, $u$, $T$ and $\alpha<\beta$ be such that $u$ is a blocking
  word for $F$ along $h_\alpha$ (resp. $h_\beta$) during time
  $T$. Then, for any $\gamma\in[\alpha,\beta]$, $u$ is also a blocking
  word along $h_\gamma$ during time $T$.
\end{lemma}
\begin{proof}
  Straightforward.
\end{proof}

\begin{prop}
  \label{prop:ulce}
  Let $\gs$ be a computable subshift and $(\az,F)$ a CA. For
  $u\in\gs$, with $I_u=|a_u,b_u|$, $a_u$ is lce and $b_u$ is rce.
\end{prop}

\begin{proof}
  We consider configurations in which $u$ is placed on the
  origin. For every $n\in\N$, the set of cells in the consequences of
  $u$ at time $n$ is computable since only a finite number of factors
  on finitely many configurations have to be considered (the factors
  can be computed because the subshift is computable).

 Consider $x_n$ the smallest integer such that $u$ is a blocking word
  during time $n$ of slope $h_{\frac{x_n}{n}}$. The sequence defined by $\alpha_1=x_1$ and $\alpha_{n+1}=max \left(\frac{x_{n+1}}{n+1},\alpha_n \right)$ is increasing and clearly computable from what was said above.  
  
  Now let $a\in I_u$. Because $u$ is a blocking word of slope
  $a$, it is also a blocking word of slope $a$ during time $n$,
%   So $(x_n, n)$ is the
%   leftmost cell such that the line segment from the right of $u$ $(|u|-1,0)$ to
%   it is a wall. Clearly, $\forall y<x_n$, $u$ is not a $\N$-blocking word for
%   the line passing through $(|u|-1,0)$ and $(y,n)$.
  which means that $\forall n\in \N,\frac{x_n}{n}\leq a$, so $\forall
  n\in \N, \alpha_n\leq a$. Therefore, the sequence $(\alpha_n)_n$
  tends toward some limit $\alpha\leq a$. As it is true for any $a\in
  I_u$, we have $\alpha\leq a_u$. Suppose for the sake of
  contradiction that there is some $b$ with $\alpha\leq b<a_u$. Then,
  by Lemma~\ref{lem:esayconvex}, $b'=\frac{a_u+b}{2}$ is such that $u$
  is a blocking word of slope $b'$ for arbitrary long time (because
  $\alpha$ and any $a\geq a_u$ are), hence a $\N$-blocking word of
  slope $b'$. We get $b'\in I_u$ which is a contradiction since
  $b'<a_u$. Thus $a_u\leq \alpha$ and finally $a_u=\alpha$. We deduce
  that $a_u$ is lce.

% *****************************************\\

  A symmetric proof shows that $b_u$ is rce. In fact the situation is
  not formally symmetric since we consider functions of the form
  $h_\alpha(n) = \lfloor\alpha n\rfloor$. However it is obvious that
  $h'_\alpha(n) = \lceil\alpha n\rceil$ is such that $h_\alpha\sim
  h'_\alpha$. Then Proposition~\ref{prop:order_on_F} allows to make a
  symmetric reasonning on functions of the form $h'_\alpha$ and still
  have a conclusion for function of the form $h_\alpha$ which are
  considered in the statement of the current Proposition.
\end{proof}

In the proof above, we actually showed that $a_u$ and $b_u$ were
directions of equicontinuity. So the set of linear directions of
equicontinuity is closed for a word.

\subsection{Bounds for $\dalk(\gs,F)$}
\begin{theorem}
\label{theo:linearbounds}
  Let $\gs$ be a computable subshift and $(\az,F)$ a CA. Let
  $\dalk(\gs,F)= \{h_{\alpha}:\alpha\in |\alpha',\alpha''|\}$.

  Both $\alpha'$ and $\alpha''$ are ce. Moreover, if
  $|\alpha',\alpha''|$ is left-closed, $\alpha'$ is lce, and if
  $|\alpha',\alpha''|$ is right-closed then $\alpha''$ is rce.
% \begin{itemize}
% \item If there exist $\alpha',\alpha'' \in\R$ such that $\dalk(\gs,F)=
%   \{h_{\alpha}:\alpha\in [\alpha',\alpha'']\}$ then $\alpha'$ is lce and
%   $\alpha''$ is rce.
% \item If there exist $\alpha',\alpha'' \in\R$ such that $\dalk(\gs,F)=
%   \{h_{\alpha}:\alpha\in ]\alpha',\alpha'']\}$ then $\alpha'$ is ce and
%   $\alpha''$ is rce.
% \item If there exist $\alpha',\alpha'' \in\R$ such that $\dalk(\gs,F)=
%   \{h_{\alpha}:\alpha\in [\alpha',\alpha''[\}$ then $\alpha'$ is lce and
%   $\alpha''$ is ce.
% \item If there exist $\alpha',\alpha'' \in\R$ such that $\dalk(\gs,F)=
%   \{h_{\alpha}:\alpha\in ]\alpha',\alpha''[\}$ then $\alpha'$ and $\alpha''$ are
%   both ce.
% \end{itemize}

\end{theorem}

\begin{proof}
  We prove it for left bounds, the proofs are similar for right
  bounds.  There are two cases: the interval is either left-closed or
  left-open. The first case follows from Proposition \ref{prop:ulce}
  since there exists a word $u$ such that the left bound of slopes of
  $\dalk(\gs,F)$ is the left bound of $I_u$.

  In the case of an open bound, we produce a sequence converging to
  it.  Suppose $\dalk(\gs,F)=]\alpha',\alpha''|$, then for every
  $i\in\N$, there exist $u_i$ such that $I_{u_i}=|x_i,y_i|$ with
  $x_i\in ]\alpha',\alpha'+1/i]$. So these $x_i$ are lce and the
  sequence $(x_i)_i$ tends to $\alpha'$. For every $i\in\N$ let
  $(y_{i,k})_k$ be a rational sequence converging to
  $x_i$. $(y_{i,i})_i$ is a rational sequence converging to $\alpha'$,
  hence $\alpha'$ is ce.
\end{proof}

Here, we consider linear directions of equicontinuity for a CA, so the intervals of admissible directions can be open or closed.\\

In the case when there is a single linear equicontinuous direction, we
have the following corollary:

\begin{corollary}
  For $\gs$ a computable subshift and a CA $(\az,F)$, if there exists
  $\alpha \in\R$ such that $\dalk(\gs,F)= \{h_{\alpha}\}$ then
  $\alpha$ is computable.
\end{corollary}
\begin{proof}
  As $\alpha$ is both a closed left and a closed right bound, it is
  left and right computably enumerable so it is computable.
\end{proof}

\subsection{Reachability}

\newcommand{\stu}{\includegraphics{figures/dpst-2}}
\newcommand{\stv}{\includegraphics{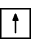}}

We prove here that any computably enumerable number is realized as a
bound for $\{\alpha|h_{\alpha}\in\dalk(\az,F)\}$ on some cellular
automaton $F$. The idea of the construction is to use the counters
described in Subsection \ref{subsec:counters} combined with
methods to obtain signals of computably enumerable slopes.

We will prove the following result:
\begin{prop}
\label{prop:reachability}
  For every ce number $0\leq \alpha\leq 1$, there exists a CA $(\az,
  F)$ and a word $u\in \A^*$ such that $I_u=]\alpha,1]$.
\end{prop}

The idea consists in constructing an area of $\stb$ states (which will be the desired set of consequences) limited to
the right by a line of slope 1, and to the left by a curve that tends
to the line of slope $\alpha$. As in the construction of the parabola
from Subsection \ref{subsec:par}, the $\stb$ signal will move up and
right, and a specific signal will be able to turn a $\stb$ state into
a $\stw$ state. We will then send the correct density of signals to
get the right slope.

As $\alpha$ is ce, there exists a Turing machine that enumerates a
sequence of rationals $(\alpha_i)_{i \in \N}$ that tends to
$\alpha$.

One cell initializes the whole construction. It creates the $\stb$
area and starts a Turing machine on its left. This machine has to
perform several tasks.

\subsubsection{Sending signals}

First it creates successive columns of size $2^i,\ i\in\N$. Columns
are delimited by a special state and contain blank $\stw$ states. In
each column, a signal bouces back and forth from one border to the
other. The time needed to go from the right border of the column to
the left border and back again to the right border is $2^{i+1}$ (see
Figure \ref{fig:columns}).

The right border of a column can be either active ($\stu$ state) or
resting ($\stv$ state). If a column is activated, it will send a
signal to the right each time the bouncing signal reaches it
(i.e. every $2^{i+1}$ steps). This signal, when emitted, passes through
all other columns and continues until it reaches a $\stb$ state. The
$\stb$ state is erased (therefore pushing the border of the $\stb$
surface one cell to the right) and the signal disappears. If a column
is resting, no signal is emitted when the bouncing signal hits the
right border.

The Turing machine constructs and initializes the columns in order to
synchronize the internal signals: the signal in column $(i+1)$ hits
the right border at a time when the one in column $i$ hits its right
border, as shown in Figure~\ref{fig:columns} (this is possible because
the period of each signal is exactly double the period of the previous
signal).  When a signal is emitted by column $i$, it has to move
through all the previous columns. Because the columns are
synchronized, it will pass through the column $0$ at time $(n\times
2^{i+1}+2^{i}-1)$ for some $n$. This means that two signals emitted by
different columns cannot be on the same diagonal (in the space-time
diagram) and therefore cannot collide.

For every $i\in\N$, when the $i$ first columns are all created, the
machine computes $\alpha_i$ with precision $2^{-(i+1)}$. The machine then
activates only the columns $k\leq i$ such that the $k$-th bit of
$\alpha_i$ is $1$. At that time the density of signals emitted by the
columns is in $\lbrack \alpha_i-2^{-(i+1)},\alpha_i+2^{-(i+1)}\rbrack$.

Moreover, as the sequence $(\alpha_i)_i$ converges to $\alpha$, for any $j\in \N$ there exists $i_j$ such that: $\forall i\geq i_j$, $\alpha$ and $\alpha_i$ share their $j$ first bits. And at some time, the Turing machine computes $\alpha_{i_j}$. From then on, the $j$ first columns are in their final state since they represent the $j$ first bits of $\alpha_i$, $i\geq i_j$. So the process converges and as these bits are the $j$ first bits of $\alpha$ too, the constructed slope tends to the desired one.  

\begin{figure}[htbp]
  \centering
  \includegraphics[width=12cm]{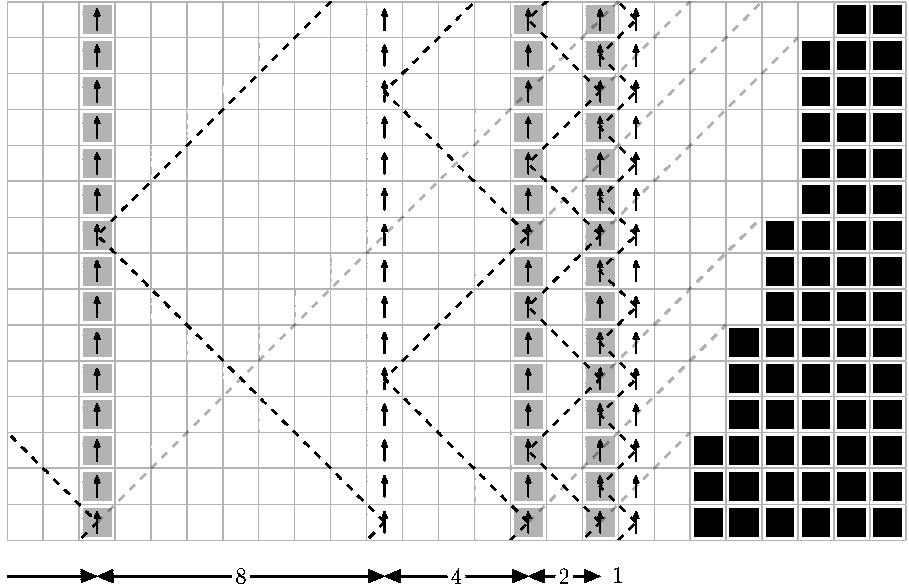}
  \caption{Columns and signals. Here the columns of width 2, 4 and 16
    are active and the signals that go through them push the black
    area at regular intervals.}
  \label{fig:columns}
\end{figure}

\subsubsection{Density of signals}
We here name density, the average number of signals emitted by a line or reaching a line during a timestep. 
Because $(\alpha_i)_i$ tends to $\alpha$, each column will be in a
permanent state (active or resting) after a long enough time. The
density of signals emitted (passing through a vertical line) by the columns will tend to $\alpha$ as
time passes. However, because signals have to reach a line that is not
vertical (the border of the $\stb$ surface), the density of signals
effectively reaching this line is less (a kind of Doppler effect). If we want to get a line of slope $\alpha$, $\alpha$ signals should reach the $\stb$ frontier at each timestep. As shown in Figure~\ref{fig:doppler}, if the density of emitted signals is some $\beta$, we have $\beta(1-\alpha)$ signals reaching the frontier, so we want $\beta(1-\alpha)=\alpha$.
And finally, we need to emit signals with a density
equal to $\frac{\alpha}{1-\alpha}$. Clearly, as $\alpha$ is ce,
$\frac{\alpha}{1-\alpha}$ is ce too so it is possible to emit the
correct density of signals.

\begin{figure}[htbp]
  \centering
  \includegraphics[width=4cm]{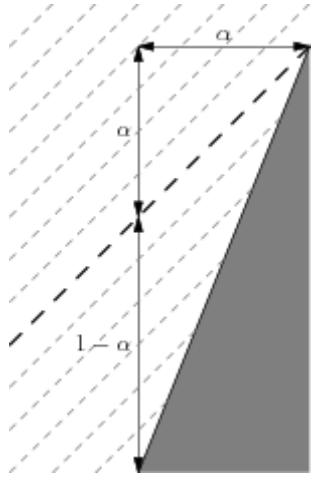}
  \caption{Signals approaching at speed 1 reach the frontier with a lesser density.}
  \label{fig:doppler}
\end{figure}

\subsubsection{Equicontinuity} 

If only the $\stb$ state can move through the borders of the columns
(and by doing so destroys them) the consequences of the initializing cell
are an area containing lines of slopes between $\alpha$ and $1$. We use again the transparency trick described in \ref{subsec:transparency}. The
left bound ($\alpha$) is not necessarily closed. If $\alpha$ is lce,
we have an increasing sequence converging to it and so the slope of
the curve bounding the $\stb$ area on the left is lower than
$\alpha$. So if $\alpha$ is lce, we can make a construction such that
the bound is closed.

\subsubsection{Right side}

It is possible to do a symmetric construction on the other side. We
can then have an area of $\stb$ states between the line of slope $-1$
and the line of slope $\alpha$. The construction must however be
slightly modified.

If we consider an automaton with a larger radius, we can have
``diagonal columns'' delimited by lines of slope 1. In this case, the
signals emitted by the columns move towards the left and they ``pull''
the black area instead of pushing it (in a way that is very similar to
the construction of the parabola in
Subsection~\ref{subsec:counters}). If we ``protect'' this construction
between counters as seen in Subsection~\ref{subsec:counters} it
becomes a set of consequences.

\subsubsection{Results}{}

With the previous constructions and cartesian products if necessary,
we have a sort of converse of Theorem~\ref{theo:linearbounds}:
\begin{theorem}
  Let $\alpha'\leq \alpha''$ be ce real numbers. There exists a CA $(\az,
  F)$ such that $\dalk(\az,F) = |\alpha', \alpha''|$. Moreover, if
  $\alpha'$ is lce the left bound can be closed: $\dalk(\az,F) =
  [\alpha', \alpha''|$. If $\alpha''$ is rce, the right bound can
  be closed.
\end{theorem}

We use the CA constructed in \ref{prop:reachability}. For each set of directions, we had a word $u$ with  exactly the desired set of consequences. It remains to prove that there is no other linear direction of equicontinuity. This can be achieved by considering a blocking word $v$ along another linear direction, and using the same kind of arguments as in \ref{prop:par}. Then $uvu$ brings a contradiction  since the consequences of $v$ should extend outside the consequences of one of the $u$. Which is not possible since $u$ is blocking.

%%%%%%%%%%%%%%%%%%%%%%%%%%%%%%%
%%%
%%%%%%%%%%%%%%%%%%%%%%%%%%%%%%%

\section{Equicontinuous dynamics: constraints and negative results}
\label{sec:limitations}

This section aims at showing that some sets cannot be consequences of
any word on a cellular automaton. We know that the consequences of a
word $u$ cannot extend to cells that never receive any information
from $u$ (except for nilpotent CA). But there are other constraints, and we study them here. For
example, a natural idea is to put a word $u$, 2 or more times on the
initial configuration. Thus, the space-time diagram contains the
consequences of $u$ and a copy of them spatially translated. Clearly,
the sites in the intersection of the consequences with its copy's have
one unique state and so, we can get relations between the states of
different sites of the consequences. In this part we give some
conditions on a set of sites to be a potential set of consequences. \\

\subsection{States in a set of consequences}

Let $F$ be a CA, $u$ be a word and $n\ge |u|\in \N$. First, we
consider configurations $c \in \az$ containing a second occurrence of
$u$ translated of $n$ cells on the right, i.e. $\exists n\ge |u|$,
${c\in[u]_0}$ and ${\sigma^n(c)\in[u]_0}$. We denote by
$\mathcal{C}(u)$ the set of consequences $\cone{F}{[u]_0}$. Now,
suppose there exist two sites $\langle x,t \rangle$ and $\langle x+n,t
\rangle$ in $\mathcal{C}(u)$, for $x\in \Z$ and $t\in \N$. If the site
$\langle x,t \rangle$ is in the state $a\in \mathcal{A}$, then,
considering the translation of $\mathcal{C}(u)$, the site $\langle
x+n,t \rangle$ is in the state $a$. So the consequences
$\mathcal{C}(u)$ impose that the sites $\langle x,t \rangle$ and
$\langle x+n,t \rangle$ are in the same state. The following
proposition generalizes this simple idea.

\begin{prop} 
\label{prop:uniformsegment}
  If, for some $l\in \N$, the set of consequences of a word
  $u\in \mathcal{A}^l$ contains the sites $\langle x,t\rangle$ and
  $\langle x+2l,t\rangle$ for some $x\in \Z$ and $t\in \N$, then all
  the sites $\langle y,t\rangle$ ($y\in \Z$) such that $\langle y,t\rangle\in \mathcal{C}(u)$, are in
  the same state for any initial configuration of $[u]_0$.
\end{prop}
\begin{proof}
  For any $y\in \Z$, either $|y-x|\geq l=|u|$, or $|y-(x+2l)|\geq l=|u|$. So considering what was proved just above, $\langle y,t\rangle$ is in the same state as either $\langle x,t\rangle$ or $\langle x+2l,t\rangle$. And the same argument proves that these both sites share their state too.
\end{proof}

\subsection{Constraints coming from periodic initial configurations}

For $n\in \N$ and $u\in \mathcal{A}^n$, we now consider initial configurations
that are periodic of period $uv$ for some $v$. We use the fact that they lead to
an ultimately periodic space-time diagram. First we consider $v$ of length $0$,
we then have $u^{\Z}$ for initial configuration. We have a ultimately periodic
diagram, and so, we can tell that the states of the consequences follow a
ultimately periodic pattern.  The following lemma illustrates a particular case
where we can show that the consequences of a word $u$ are \emph{eventually spatially uniform}:
all sites in the consequences of $u$ at any given large enough time are in the
same state.

\begin{lemma}
  Suppose that for some word $u$ of length $n$, all the space-time diagrams with initial
  configurations in $S=\{(uv)^{\Z} : |v|=n^2\}$ have identical periodic
  part. Then the consequences of $u$ are eventually spatially uniform.
  \label{lem:uniformbyline}
\end{lemma}
\begin{proof}
  Consider periodic configurations of the form $(uv)^{\Z}$ with $v=u^n$ and $v=a(ua)^{n-1}$ ($a\in
  \mathcal{A}$), respectively. They respectively have
  (spatial) period lengths $|u|$ and $|u|+1$, so their space-time diagrams
  too. Since they are in $S$, the equality of their periodic part implies that
  they have both spatial periods $|u|$ and $|u|+1$. It follows that it has
  period $1$, so the periodic part is eventually spatially uniform, and in particular, the
  consequences of $u$ are eventually spatially uniform.
\end{proof}

We will study below examples of sets of consequences where the lemma applies. As
we want to show the equality of two periodic space-time diagrams, we only need
to show it on one spatial period at some time. So we will only show the equality
of both configurations on a
segment of length $n^2+n$.\\

We denote by $P_{uv}$ and $P_{uw}$ the space-time diagrams with periodic initial
configurations of periods $uv$ and $uw$ ($|v|=|w|=n^2$). In the periodic part of
them, the spatial period will be $n^2+n$, and the temporal periods will be
$T_{uv}$ and $T_{uw}$. So, a common period can be defined by vectors $(n^2+n,0)$
and $(0,T)$ where $T=T_{uv}T_{uw}$.

\subsubsection{Parabola}

We now consider a word $u$ whose set of consequences draws a discrete
parabola. The definition of a parabola that we will use here is that it is a
sequence of vertical segments of increasing lengths, and translated by 1 to the
right compared to the previous one. More formally, we suppose that the set of
consequences of $u$ verifies the
following: \[\mathcal{C}(u)\supset\left\{\langle x,y\rangle : f(x)\leq y
  <f(x+1)\right\},\] for some polynomial function $f$
which is strictly increasing on $\N$, and such that ${x\mapsto f(x+1)-f(x)}$ is strictly increasing too.\\

\begin{prop}
  \label{prop:uniformparabola}
  If the consequences of some word $u$ contains a parabola in the above sense
  then they are eventually spatially uniform.
\end{prop}

\begin{proof}
  Let's suppose $T_0$ is the smallest integer such that after $T_0$, both
  $P_{uv}$ and $P_{uw}$ are in their periodic part. Now we take $x'$ such that
  $T\leq f(x'+1)-f(x')$. We take $t>T_0$ and $x>x'$ such that $\langle x,t
  \rangle \in \mathcal{C}(u)$. These $x$ and $t$ exist thanks to the definition
  of $\mathcal{C}(u)$, and they are large
  enough to be in the periodic part of both diagrams. \\

  We now show that $P_{uv}$ and $P_{uw}$ coincide on the sites $\langle x+k,t
  \rangle$ for $0 \leq k < n^2+n$ and lemma~\ref{lem:uniformbyline} concludes.
  To do this we show that, for any $k$, there exists a site $s\in\mathcal{C}(u)$
  such that \[P_{uv}(\langle x+k,t \rangle)=P_{uv}(s)=P_{uw}(s)=P_{uw}(\langle
  x+k,t \rangle).\]

  To find this site $s$, we consider the temporal periodicity: at $t+mT$, the
  states will be the same as at $t$ in $P_{uv}$ (period $T_{uv}$) and in
  $P_{uw}$ (period $T_{uw}$) for all $m$. From that, it is sufficient to find
  for each $1\leq k <n^2+n$ an $m$ such that $s=\langle x+k,t+mT \rangle \in
  \mathcal{C}(u)$. As we have taken $x>x'$ and with the properties of $f$, we
  know that $f(x+k+1)-f(x+k)\geq T$ and there are more than $T$ consecutive
  sites of abscissa $x+k$ in $\mathcal{C}(u)$, so such an $m$ exists for all
  $k$. 
\end{proof}

The proposition applies to the examples constructed in the previous section. It
shows that the non-linear set of consequences like parabolas are obtained at the
price of spacial uniformity.

\subsubsection{Non-periodic walls}

Now we consider a wall $u$ \along{} some $h$ (with bounded variations) where
$\forall \alpha \in \Q,h\nsim h_{\alpha}$. In this case again, the consequences
of $u$ are eventually spatially uniform.

\begin{prop}\label{aperwall} If there exists a blocking word $u$
\along{} $h$ such that $\forall \alpha \in \Q,h\nsim h_{\alpha}$, then
the consequences of $u$ are eventually spatially uniform.
\end{prop}

\begin{proof} We once more consider the set of configurations $\{(uv)^{\Z} :
  |v|=|u|^2\}$, and prove that the periodic parts of the generated space-time
  diagrams are all identical. Then lemma~\ref{lem:uniformbyline} concludes. Let
  $v$ and $w$ be arbitrary words of length $|u|^2$. If some site $\langle
  x,t\rangle$ has different states in the space-time diagrams $P_{uv}$ and
  $P_{uw}$, there must be another site $\langle y,t-1\rangle$ with different
  states too. And $x-y$ must be bounded by the radius of the automaton. So if
  there are differences in the periodic parts of $P_{uv}$ and $P_{uw}$, we have
  a sequence of sites $s_n=\langle x_n,n\rangle$ with different states in both
  diagrams. And for all $n>0$, $x_n-x_{n-1}$ is bounded by the radius of the
  automaton, but a wall is by definition larger than the radius of the
  automaton, so this sequence can't cross a wall.
% \\ As we have an infinite
%   sequence, there exist $p$ and $p'$ such that $s_p$ is the translated of
%   $s_{p'}$ by a vector of periodicity of both diagrams $(a,b)$. If $m=p-p'$, we
%   can now construct the sequence of sites $s'_n$ such that:
%   \begin{itemize}
%   \item if $n \leq p$ then $s'_n=s_n$
%   \item if $n>p$ then $s'_n=\langle x_{p+\left((n-p) mod (p-p')\right)}+\lfloor
%     \frac{n-p'}{p-p'}\rfloor a,n\rangle$
%   \end{itemize} Thanks to the periodicity, $s'_n$ has still different states in
%   both diagrams.
  Since the set of sites where $P_{uv}$ and $P_{uw}$ differ is periodic after
  some time, the sequence $(s_n)$ can also be chosen ultimately periodic.  So we
  have a sequence that can't cross a wall either, and that is ultimately
  periodic. But we can have the same sequence translated on the other side of
  the wall. So the wall is between two sequences with the same (ultimate)
  period. We can associate a slope $\alpha$ to this period, and thus we have
  $h\sim h_{\alpha}$ for some $\alpha \in \Q$. Which is a contradiction.
\end{proof}

\subsection{Reversible CA} 

We now restrict to reversible CA. Arguments developed earlier have stronger
consequences in this case.

\begin{prop} 
  On a reversible CA, if the set of consequences of a word $u$ of length $n$
  contains all sites $\langle x+k,t\rangle$ for $0\leq k \leq 2n-1$ and some $x$
  and $t$, then $u$ is uniform and $\mathcal{C}(u)$ are eventually spatially uniform.
\end{prop}
\begin{proof}
  Proposition~\ref{prop:uniformsegment} let us conclude that a long segment of
  sites in the consequences of $u$ are in the same state. Actually, this segment
  is twice as long as the initial word $u$, so if we start with the initial
  configuration $u^{\Z}$, we obtain an uniform configuration after some
  time. But we have here a reversible CA, so the initial configuration had to be
  uniform too. And the consequences are also eventually uniform by line.
\end{proof}

Remark that a uniform segment of length $n$ in the consequences is sufficient to
conclude in the above proposition.

In the case of reversible CA, the hypothesis of lemma~\ref{lem:uniformbyline}
can never be satisfied: indeed two different initial configurations cannot have
identical images. Therefore, since proofs of
propositions~\ref{prop:uniformparabola} and \ref{aperwall} consist in showing
that lemma~\ref{lem:uniformbyline} applies, we deduce that hypothesis of each of
these proposition are never satisfied by any reversible CA. This is summarized
by the following theorem.

\begin{theorem}
  \label{exrev}
  \label{blorev}
  Consider any reversible cellular automaton. Then:
  \begin{itemize}
  \item no word can contain a parabola in its set of consequences;
  \item there can be no blocking word \along{} $h$ such that $\forall \alpha \in
    \Q,h\nsim h_{\alpha}$.
  \end{itemize}
\end{theorem}

\subsubsection {Negative consequences}

We now consider also negative times. We have the following result.

\begin{prop} \label{negcons} In a reversible cellular automaton, if the set of positive consequences of a word $u$ contains a line $D$ of rational slope then $D$ is also in the
  negative consequences of $u$.
\end{prop}

\begin{proof} Let $\langle x,t \rangle$ be a site on $D$ with $t<0$. We suppose the site $\langle x,t
  \rangle$ is not a consequence of $u$. As the reverse of a cellular automaton
  is still a CA, we can take two \emph{extensions} $u_1$ and $u_2$ of $u$ that
  force $\langle x,t \rangle$ in different states. We take $|u_1|=|u_2|=n$. As
  the line $D$ is of rational slope, its representation on the discrete plane is
  periodic, let's call $(a,b)$ a vector of periodicity. Now we consider the
  space-time diagrams obtained with periodic configurations of period $u_1$ and
  $u_2$. They are both periodic and we can find some common period $(n,T)$. We can choose $T>t$ without loss of generality.\\

  As $(n,0)$ and $(0,T)$ are periods of the both space-time diagrams, $\langle x+nTa,t+nTb
  \rangle$ and $\langle x,t \rangle$ have the same state.  As $\langle x,t
  \rangle$ belongs to $D$, $\langle x+nTa,t+nTb \rangle$ belongs to $D$ too. We
  have $T>t$ so $t+nTb>0$ and $\langle x+nTa,t+nTb \rangle$ is forced in a
  unique state in both diagrams because it is in the consequences of $u$. So
  $\langle x,t \rangle$ is in the same state in both diagrams too. Which is a
  contradiction with the fact that $u_1$ and $u_2$ force $\langle x,t \rangle$
  in different states.
\end{proof}

Combining \ref{blorev} and \ref{negcons}, we have the following result:

\begin{theorem} On a reversible cellular automaton, every
  $\N$-blocking word along $h\in\Fon$ is also a $\Z$-blocking word along
  $h$.
\end{theorem}

%%%%%%%%%%%%%%%%%%%%%%%%%%%%%%%
%%%
%%%%%%%%%%%%%%%%%%%%%%%%%%%%%%%

\section{Future directions}

We believe that the following research directions are worth being considered.

\begin{itemize}
\item What are the possible shapes of $\blz(\gs,F)$? In \cite{Sablik-2008}, it
  is shown that the set of slopes of linear directions of expansivity is an
  interval. How does this generalize to arbitrary curves?
\item We have shown that reversibility adds a strong constraint on possible sets
  of consequences of words, but what kind of restrictions impose the property of
  being surjective?
\item What are the precise links between $\aln(\gs,F)$ and $\daln(\gs,F)$? When
  $\aln(\gs,F)=\{h\in\Fon : h_1 \prec h\prec h_2\}$ is there a connection
  between
  $\left[\liminf_{n\to\infty}\frac{h_1(n)}{n};\limsup_{n\to\infty}\frac{h_2(n)}{n}\right]$
  and $\daln(\gs,F)$?
\item The construction techniques developed to obtain the parabola as an
  equicontinuous curve are very general and there is no doubt that a wide
  family of curves can be obtained this way. Can we precisely characterize
  admissible curves of equicontinuity as we did for slopes of linear
  equicontinuous directions?
\item We believe that the measure-theoretic point of view should be considered
  together with this directional dynamics framework. For instance, the
  construction of corollary~\ref{cor:gardenedenwall} has interresting
  measure-theoretic properties.
\item This paper is limited to dimension $1$. The topological dynamics of higher
  dimensional CA is more complex \cite{Sablik-Theyssier-2008} and the
  directional framework is a natural tool to express rich dynamics occuring in
  higher dimension.
\end{itemize}

% \item Mieux comprendre $\alk(\gs,F)$. Notament la fronti\`ere!!! Que
% pouvons nous construire?

% \item Est-il n\'ecessaire de distinguer le cas des AC bijectifs pour
% les points d'equicontinuite? En d'autre terme, il faut trouver des
% resultats pour le cas bijectif. Quelques questions en vrac:
%   \begin{itemize}
%   \item Si $F$ est bijectif et a un point de $\N$-equicontinuite,
% est ce qu'il a un point de $\Z$-equicontinuite?

% \item Il est facile de construire un AC tel
%   que $$\blz(\az,F)=\bigcup_{k\in[1,n]}]\alpha_k^{l},\alpha_k^{r}[,$$
%   o\`u
%   $-\infty=\alpha_1^{l}<\alpha_1^{r}<\alpha_2^{l}\leq\alpha_2^{r}<\cdots<\alpha_{n-1}^{l}\leq\alpha_{n-1}^{r}<\alpha_n^{l}<\alpha_n^{r}=+\infty$
%   avec $\alpha_k^{l}$ et $\alpha_k^{r}$ in $\Q$ pour tout
%   $k\in[1,n]$. Cependant, peut-on construire d'autres exemples tel que
%   $\blz(\gs,F)$ a une autre forme.
%   \end{itemize}

\def\ocirc#1{\ifmmode\setbox0=\hbox{$#1$}\dimen0=\ht0 \advance\dimen0
  by1pt\rlap{\hbox to\wd0{\hss\raise\dimen0
  \hbox{\hskip.2em$\scriptscriptstyle\circ$}\hss}}#1\else {\accent"17 #1}\fi}

\end{document}